\documentclass[11pt]{article}

\usepackage{amssymb,amsthm,xspace}
\usepackage{graphicx}
\graphicspath{{figures/}{}}
\usepackage[T1]{fontenc}
\usepackage{graphicx}%[pdftex]
\usepackage{color,verbatim}
\usepackage{multirow}
\usepackage{amsmath,amsfonts,dsfont}
\usepackage{pdfpages,hyperref,cite,url}
\usepackage{fancybox}
\usepackage{enumerate}
\usepackage[margin=1in]{geometry}
\usepackage{authblk}

\usepackage{url}

\newcommand{\techrep}[2]{#1}

\newcommand{\gls}{\ensuremath{\mathtt{GLS}}\xspace}
\newcommand{\opt}{\ensuremath{\mathtt{opt}}}
\newcommand{\NE}{\mathtt{NE}}
\newcommand{\pos}{\mathtt{PoS}}
\newcommand{\poa}{\mathtt{PoA}}

\newcommand{\R}{\mathbb R}
\newcommand{\E}{\mathbb E}
\newcommand{\T}{T}
\newcommand{\prcs}{A}
\newcommand{\cov}{V}
\newcommand{\ndim}{d}

\newcommand{\SC}{\ensuremath{C}}

\DeclareMathOperator*{\argmin}{arg\,min}
\DeclareMathOperator{\trace}{trace}

\DeclareMathOperator{\diag}{diag}
\DeclareMathOperator{\dom}{dom}

\newcommand{\xb}{\boldsymbol{x}}
\newcommand{\yb}{\boldsymbol{y}}

\newcommand{\betab}{\boldsymbol{\beta}}

\newcommand{\lambdab}{\boldsymbol{\lambda}}

\newcommand{\lambdaeq}{\lambda^{*}}
\newcommand{\lambdabeq}{\lambdab^{*}}

\newcommand{\deltab}{\boldsymbol{\delta}}

\newcommand\pone{{{p}_{\textrm{min}}}}
\newcommand\ptwo{{{p}_{\textrm{max}}}}

\newtheorem{theorem}{Theorem}%[section]
\newtheorem{corollary}{Corollary}%[section]
\newtheorem{lemma}{Lemma}%[section]
\newtheorem{assumption}{Assumption} 

\begin{document}

\title{Linear Regression from Strategic Data Sources
\thanks{This paper is an extended version of ``Linear regression as a non-cooperative game'', by Ioannidis and Loiseau~\cite{ioannidis2013linear}.} 
\thanks{This work was supported by the French National Research Agency through the ``Investissements d'avenir'' program (ANR-15-IDEX-02) and through grant ANR-16-TERC0012; by the DGA; by the Alexander von Humboldt Foundation; and by MIAI @ Grenoble-Alpes. Stratis Ioannidis acknowledges support from NSF grants  CCF-1750539 and  CNS-1717213. We thank the editor and the three anonymous reviewers for their particularly thoughtful comments and feedback, that significantly improved the paper.}}
\date{\today}
\author[1]{Nicolas Gast}
\author[2]{Stratis Ioannidis}
\author[1,3]{Patrick Loiseau}
\author[1]{Benjamin Roussillon}
\affil[1]{Univ. Grenoble Alpes, Inria, CNRS, Grenoble INP, LIG}
\affil[2]{Northeastern University}
\affil[3]{Max-Planck Institute for Software Systems (MPI-SWS)}

\maketitle

\begin{abstract}
	Linear regression is a fundamental building block of statistical data analysis. It amounts to estimating the parameters of a linear model that maps input features to corresponding outputs. In the classical setting where the precision of each data point is fixed, the famous Aitken/Gauss-Markov theorem in statistics states that generalized least squares  (GLS) is a so-called ``Best Linear Unbiased Estimator'' (BLUE). 
In modern data science, however, one often faces \emph{strategic data
	sources}, namely, individuals who incur a cost for
providing high-precision data. For instance, this is the case for personal data, whose	revelation may affect an individual's privacy---which can be modeled as a cost---or in applications such as recommender systems, where producing an accurate estimate entails effort. 

In this paper, we study a setting in which features are public but
individuals choose the precision of the outputs they
reveal to an analyst. We assume that the analyst performs linear regression on this dataset, and individuals benefit from the outcome of this estimation.
We model this scenario as a game where
individuals minimize a cost comprising two components: (a) an (agent-specific)
disclosure cost for providing high-precision data; and (b) a (global)
estimation cost representing the inaccuracy in the linear model
estimate. In this game, the linear model estimate is a public good
that benefits all individuals. We establish that this game has a unique
non-trivial Nash equilibrium. We study the efficiency of this
equilibrium and we prove tight bounds on the price of stability for a
large class of disclosure and estimation costs.  Finally, we study the
estimator accuracy achieved at equilibrium. We show that,  in general, Aitken's
theorem does not hold under strategic data sources, though it does
hold if individuals have identical disclosure costs (up to a
multiplicative factor). When individuals have non-identical costs, we
derive a bound on the improvement of the equilibrium estimation cost that
can be achieved by deviating from GLS, under mild assumptions on the
disclosure cost functions.

\end{abstract}

\section{Introduction}
%!TEX root = techreport.tex

The statistical analysis of data is a cornerstone of many scientific disciplines. The core problem of estimating the parameters of a model is classic, and is  well understood in the standard setting in which any noise or distortions present in the data are exogenous. Naturally, the quality of the data, as captured by such noise or distortions, is key to an estimator's accuracy. 
In many instances, however, obtaining high quality data may be associated with a cost incurred by the data source. For example, this is the case when the data is of a personal nature and provided by privacy-conscious individuals. The quality of the data provided in this case can come at a cost of a violation of privacy \cite{ghosh-roth:privacy-auction,roth-liggett,Cummings15a}. The desire for privacy incentivizes individuals to obfuscate their private information, or, in the extreme, altogether refrain from any disclosure.  An additional setting in which high-quality data may come at a cost is when the data quality depends on effort exerted \cite{Cai15a,Westenbroek19a}, i.e., improved quality is the result of increased effort expended by the data source. This setting  naturally arises in, e.g., open collaboration projects, such as wikipedia, but also in online recommender systems, where individuals need to exert effort (complete surveys, click ``like'' buttons, etc.) to disclose their preferences. Just as in the privacy case, a data source may choose to not exert the effort required to produce high-quality, low-noise responses or, in the extreme,  altogether refrain from reporting anything meaningful.

In this setting, where providing high-quality data comes at a cost, it makes sense to consider strategic behavior among data sources. In particular, one should ask: why would strategic data sources provide any data at all? The existing literature focuses either on the case where individuals receive a monetary compensation to provide data \cite{ghosh-roth:privacy-auction,Cummings15a,Cai15a,Luo15a,Westenbroek19a,Abernethy15a}, or on the case where individuals care about the quality of the estimation but only w.r.t.~predictions over their own features \cite{dekel2010incentive,Chen18a,perote2004}.
We consider another possibility, namely, that \emph{globally successful data analysis may also provide a utility to the individuals from which the data is collected}. This is evident in medical studies: an experiment may lead to the discovery of a treatment for a disease, from which an experiment subject may benefit.  In the case of  recommender systems, users may indirectly benefit from overall service improvements, as data disclosed may lead to, e.g., improved product recommendations or better-targeted advertising. Similarly, open collaboration projects, by their nature, implicitly assume a common underlying utility, linked to the success of the collaboration. If such benefits outweigh associated privacy or effort costs, individuals may consent to the collection and analysis of high-quality data, e.g., by participating in a clinical trial, completing a survey, or disclosing their preferences in a recommender service.

In this paper, we approach the above issues through a non-cooperative game, focusing on the basic statistical analysis task of \emph{linear regression}. We consider the following formal setting. A set of individuals $i\in \{1,\ldots,n\}$ participate in an experiment in which they are asked to provide data to an analyst. Each individual $i$ is associated with a feature vector $\xb_i\in \R^\ndim$, capturing public information such as age, gender, etc., and possesses a private variable $y_i\in \R$---e.g., the true answer to a survey, the outcome of a medical test, or how much they like a product. The analyst wishes to perform linear regression over the data, i.e., compute a vector $\betab\in \R^\ndim$ such that:
\begin{equation*}
 y_i \approx \betab^\T \xb_i, \qquad \text{for all }i \in \{1,\ldots,n\}. 
\end{equation*}

We assume that individuals benefit from the correct estimation of $\betab$: for example, if the analyst learns $\betab$, an individual may benefit due to, e.g., better medical treatment, improved recommendations, etc. However,
individuals \emph{do not} disclose their true private variables to the analyst. Instead, they provide a perturbed version $\tilde{y}_i$, constructed by \emph{adding noise} to the private variable $y_i$.  
This is because there is a cost associated with the disclosure of the private variable: the higher the variance of the noise an individual adds, the lower the cost (e.g., due to privacy violation or effort exerted) she incurs. On the other hand, high noise variance lowers the accuracy of the analyst's estimate of $\betab$, the linear model computed in aggregate across multiple individuals. As such, the individuals need to strike a balance between the  cost they incur through disclosure and the utility they accrue from accurate model prediction. 

We make the following contributions: 
\begin{itemize} 
\item[$(i)$] We model interactions between data sources as a non-cooperative game, in which each data source selects the precision (i.e., the inverse of the noise variance)  of the private variable she  strategically discloses. A data source's decision minimizes a cost function comprising two components: (a) a \emph{disclosure cost}, that is an increasing function of the chosen precision, and (b) an \emph{estimation cost}, that decreases as the accuracy of the analyst's estimation of $\betab$ increases. Formally, the estimation cost is a function of the covariance matrix of  the estimate of $\betab$. 
\item[$(ii)$] We characterize the Nash equilibria of the above game. In particular, we show that it is a potential game and that, under appropriate assumptions on the disclosure and estimation costs, there exists a unique pure Nash equilibrium at which individual costs are finite.
\item[$(iii)$] Armed with this result, we determine the game's efficiency, providing bounds for the price of stability  for several cases of  disclosure and estimation costs. 
\item[$(iv)$] Finally, we turn our attention to the analyst's estimation algorithm. In the presence of non-strategic data sources, the Aitken theorem\footnote{Also known as the Gauss-Markov theorem in the special case of ordinary least squares.} states that generalized least squares estimation yields minimal covariance among linear unbiased estimators. We challenge this theorem in the context of strategic data sources and obtain both positive and negative results:
  \begin{itemize}
  \item[a.] We show that, in general, an equivalent of Aitken theorem no longer holds under strategic data sources. We exhibit a series of counter-examples that show that, when data sources have non-identical disclosure cost functions, there exist linear estimators that lead to more efficient equilibria  (i.e., that attain more accurate estimates at equilibrium) than generalized least squares. 
  \item[b.] We show that, when agents have monomial cost functions with identical exponents (but with possibly non-identical multiplicative factors), then generalized least squares is optimal  among the class of unbiased linear estimators even in the strategic setting: it indeed yields the most accurate estimate at equilibrium.
  \item[c.] Finally, even when generalized least squares is suboptimal, under mild assumptions on the disclosure cost functions, we show that the improvement of the equilibrium estimation cost that can be achieved by deviating from generalized least squares is bounded by a factor that depends on the heterogeneity of  data source disclosure costs. 
  \end{itemize}
Our results imply that the optimality of the generalized least squares estimator does not persist if data sources strategically choose the variance of their data.
\end{itemize}

More broadly, we model the outcome of a statistical data analysis---the estimator's accuracy---as a \emph{public good}: data sources contribute to the public good by providing high-precision data at a disclosure cost and benefit in return from the global estimator's accuracy. As is natural in such public good games, we find that data sources typically contribute a level of data precision at equilibrium that is suboptimal from the social welfare perspective (i.e., there is partial free-riding). 
More surprisingly though, we also find that under such strategic data sources, standard statistics results are challenged: it is sometimes possible to deviate from the standard estimator to increase the public good provision at equilibrium, without involving any monetary compensation.

The remainder of this paper is organized as follows. We present related work in Section~\ref{related}. Section~\ref{model} contains a review of linear regression and the definition of our non-cooperative game. We characterize Nash equilibria in Section~\ref{nash} and discuss their efficiency in Section~\ref{stability}. Our results on the optimality (or non-optimality) of generalized least squares are in Section~\ref{gauss}, and our conclusions in Section~\ref{conclusions}. All proofs are relegated \techrep{to appendices.}{to our technical report \cite{techrep}.}

\section{Related Work}\label{related}
%!TEX root = techreport.tex

\paragraph{Data Perturbation for Privacy.}

Perturbing a dataset before submitting it as input to a data mining algorithm has a long history in privacy-preserving data-mining (see, e.g., \cite{vaidya2005privacy,domingo2008survey}). Independent of an algorithm, early research focused on perturbing a dataset prior to its public release \cite{traub1984statistical,duncan2000optimal}. Perturbations  tailored to specific data mining tasks have also been studied in the context of, e.g., reconstructing  the original distribution of the underlying data \cite{agrawal2000privacy}, building decision trees \cite{agrawal2000privacy}, clustering \cite{oliveira2003privacy}, and association rule mining \cite{atallah1999disclosure}. We approach such perturbation techniques via a non-cooperative setting, where individuals strategically choose the perturbation to their data. 

The above setting differs from the framework of $\epsilon$-differential privacy \cite{Dwork06,kifer2012private}, which has also been studied from the perspective of mechanism design \cite{approximatemechanismdesign}. In differential privacy, noise is added to the \emph{output} of a computation, which is subsequently publicly released. 
The analyst performing the computation is a priori trusted; as such, individuals submit unadulterated inputs. Several works study mechanisms incentivizing data disclosure under costs quantified by differential privacy~\cite{roth-liggett,Dandekar12a,ghosh-roth:privacy-auction,Cummings15a},   whereby individuals are compensated for the privacy cost they incur. In contrast, we do not assume that the analyst is trusted, which motivates input perturbation. Such input perturbations also correspond to the more recently studied notion of local differential privacy \cite{Duchi13a,Kairouz16a}, though such studies focus on the privacy/utility tradeoff, ignoring the strategic aspect of the input perturbation.

\paragraph{Strategic Data Sources and Data Elicitation.} A few recent works consider settings where  sources may choose their effort when generating data \cite{Cai15a,Luo15a,Westenbroek19a} or have heterogeneous costs due to disclosure \cite{Abernethy15a}. In all these works, data sources are assumed to maximize the payment received (minus cost of effort). We note that in Cai \emph{et al.}~\cite{Cai15a} and Westenbroek \emph{et al.}~\cite{Westenbroek19a}, which are the closest to our work, the disclosure costs of data sources are linear in the exerted effort, whereas we use more general convex costs. 
 The data elicitation literature also includes related problems, in which one tries to incentivize an expert to truthfully reveal her prediction of an event, typically using scoring rules \cite{Frongillo15a} (see also the literature on incentives in crowdsourcing \cite{Dasgupta13a}).

A number of papers also consider data acquisition in sequential settings \cite{Abernethy15a,Chen16a,Chen18b}.
All this literature, however, considers agents that aim to maximize  the payment received but are insensitive to the quality of the learning result. Moreover, agents aim to optimize payments while the learning algorithm is fixed; the only exceptions to the latter are \cite{Chorppath13a,Caragiannis16a}, which are restricted to the case of averaging and do not consider learning tasks such as regression. In contrast, in this work, we do not involve payments but  assume that data sources benefit from the result of the learning algorithm.

\paragraph{Strategy-Proof Statistical Inference.}

Several papers study regression from the perspective of mechanism design, whereby  private variables are directly reported by strategic agents. In particular, Dekel \emph{et al.}~\cite{dekel2010incentive} consider a broad class of regression problems in which data sources may misreport their private values, and determine loss functions under which empirical risk minimization is group strategyproof. The special case of linear regression is also treated, albeit in a more restricted setting, by Perote and Perote-Pe\~na \cite{perote2004}, who identify more general strategyproof mechanisms for the 2-dimensional case. More recently, Chen \emph{et al.}~\cite{Chen18a} consider a similar setting and propose a family of group strategyproof regression mechanisms for any dimension, extending the results of both \cite{dekel2010incentive} and \cite{perote2004}. As in our paper, those works assume that the independent variables (the $\xb_i$'s) are public information and mostly look at mechanism design without money. Several papers also analyze similar problems in the case of classification \cite{Meir12a,Hardt16a} (see also a recent variant in \cite{Ben-Porat19a}).

 In contrast to this line of research, our work assumes that the analyst uses a fixed algorithm (GLS or a linear unbiased estimator). We also assume that individuals choose the precision of the data reported (and not directly the reported value), and no design is required as the chosen precision is assumed known. The main difference, however, is conceptual: in \cite{perote2004,dekel2010incentive,Chen18a} agents care about the estimation error on their instance only, whereas we assume that agents benefit equally from the global downstream effects of an accurate predictor. We consider noise addition as a non-cooperative game, focusing on pure Nash equilibrium as a solution concept and studying its efficiency.

\paragraph{Non-Cooperative Regression Games.}
Closer to our setting, Hossain and Shah~\cite{Hossain19a} also consider the pure Nash equilibrium as a solution concept in regression games and investigate its efficiency, albeit in a model closer to \cite{dekel2010incentive, Chen18a} than to ours. Interestingly, this work considers the mean squared error, a standard quantity to measure a model's quality in linear regression, instead of our estimation cost based on the covariance matrix. Our estimation cost, however, includes a somewhat broader family of functions satisfying mild assumptions (see Assumptions~\ref{estimationassumption} and \ref{homogeneityf}).

Our paper is an extended version of ``Linear regression as a non-cooperative game'', by Ioannidis and Loiseau~\cite{ioannidis2013linear}.
We generalize and tighten the price of stability results, and correct Theorem~6 of \cite{ioannidis2013linear} that stated that GLS is optimal for all cost functions. We show that this result is not true in general but that (i) it holds when agents have identical cost functions (up to a multiplicative constant) and (ii) the sub-optimality of GLS in the case of non-identical disclosure cost functions can be bounded under mild assumptions.

\paragraph{Experimental Design.}

In classic experimental design \cite{pukelsheim2006optimal,atkinson2007optimum,boyd2004convex}, an analyst observes the public features of a set of experiments, and determines which experiments to conduct with the objective of learning a linear model (from non-strategic sources). The quality of an estimated model is quantified through a scalarization of its variance \cite{boyd2004convex}. As discussed in Section~\ref{sec.game}, many such scalarizations are used in the literature, including  the so-called A-optimality, C-optimality, and D-optimality criteria we define in~\eqref{specific}. We focus on non-negative scalarizations, to ensure meaningful notions of efficiency (as determined by the price of stability in Section~\ref{stability}). Among these classic scalarizations, A-optimality and E-optimality  satisfy both our technical assumptions (Assumptions \ref{estimationassumption} and \ref{homogeneityf}), while D-optimality satisfies only our convexity assumption (Assumption~\ref{estimationassumption}). As we note in Section~\ref{sec.game}, convexity implies that the information gain (i.e., the cost reduction) due to new experiments is a submodular function. This has implications about mechanism design as well. For example, Horel \emph{et al.}~\cite{horel2013budget} exploit this to produce a polytime mechanism with approximation guarantees for a version of the experimental design problem in which subjects report their private values truthfully, but may lie about the costs they require for their participation.

\paragraph{Public Good Provision Problems.}
We finally note that our model has analogies to models used in \emph{public good} provision problems (see, e.g., \cite{Morgan00a} and references therein). Indeed, the estimate variance reduction can be seen as a public good in that, when a source contributes data, all other sources in the game benefit. As is standard in such literature, our model assumes that the disclosure costs (corresponding to provision costs in public good problems) and the estimation cost (mapping to the public good benefit) are fully separable. This analogy is pushed further in \cite{Chessa15c,Chessa17b} where the authors propose a simple mechanism to increase the provision of the public good at equilibrium in the simple case of averaging (corresponding to the standard public good framework).

\section{Model Description}\label{model}
%!TEX root = techreport.tex

In this section, we give a detailed description of our linear regression game and the agents involved. Before discussing strategic considerations, we give a brief technical review of linear models, as well as key properties of least squares estimators; all related results presented here are classic (see, e.g., \cite{Hastie09a}).

\paragraph{Notational conventions.} 

We use boldface type (e.g., $\xb$, $\yb$, $\betab$) to denote vectors (all vectors are column vectors), and capital letters (e.g., $A$, $B$, $V$) to denote matrices. As usual, we denote by $S_{+}^\ndim,S_{++}^\ndim \subset \R^{\ndim \times \ndim}$ the sets of (symmetric) positive semidefinite (PSD) and positive definite matrices of size $d\times d$, respectively. 
For two positive semidefinite matrices $A,B\in S_+^\ndim $, we write that $A\succeq B$ if $A-B\in S_+^\ndim $;  
recall that $\succeq$ defines a partial order over $S_+^\ndim $. We say that $F:S_+^\ndim \to\R$ is non-decreasing in the positive semidefinite order if $F(A)\geq F(A')$ for any two  $A, A' \in S^d_{+}$ such that $A\succeq A'$.  Moreover, we say that a matrix-valued function $F:\R^n\to S_+^d$ is \emph{matrix convex}  if $\alpha F(\lambdab)+ (1-\alpha)F(\lambdab') \succeq F(\alpha \lambdab+(1-\alpha) \lambdab')$ for all $\alpha\in [0,1]$ and $\lambdab,\lambdab'\in \R^n$.

\subsection{Linear Models}

Consider a set of $n$ data sources, henceforth referred to as \emph{agents}, denoted by  $N\equiv\{1, \cdots, n\}$. Each agent $i \in N$ is associated with a vector $\xb_i\in \R^\ndim $, the \emph{feature vector}, which is public; for example, this vector may correspond to publicly available demographic information about the agent, such as age, gender, etc. Each  $i\in N$ is also associated with a private variable $y_i\in \R$; for example, this may express the likelihood that this agent contracts a disease, the concentration of a substance in her blood, or a true answer to a survey by that agent.

We assume that the agent's private variable $y_i$ is a linear function of her public features $\xb_i$. In particular,  there exists a vector $\betab\in\R^\ndim $, the \emph{model}, such that the private variables are given by 
\begin{align}\label{linear}
y_i = \betab^T\xb_i +\epsilon_i,\quad \text{for all}~i\in N,
\end{align}
where the ``inherent noise'' variables $\{\epsilon_i\}_{i\in N}$ are
i.i.d.\footnote{To ease notation, we assume that the
  variance of the inherent noise $\sigma^2$ is identical for all agents, but all results of the
  paper remain valid if we allow this variance (or equivalently the upper bound on the precision $1/\sigma^2$, see below) to depend on the
  identity of Agent $i$.}~zero-mean random variables in $\R$ with
finite variance $\sigma^2$. We make no further assumptions on the
noise; in particular, we \emph{do not} assume it is Gaussian.

An \emph{analyst} wishes to observe the $y_i$'s and infer the model $\betab\in \R^\ndim $. This type of inference is ubiquitous in experimental sciences, and has a variety of applications. For example, the magnitude of  $\betab$'s coordinates captures the effect that features  (e.g., age or weight) have on $y_i$ (e.g., the propensity to get a disease), while the sign of a coordinate captures positive or negative correlation. Knowing $\betab$ can also aid in prediction: an estimate of private variable $y\in \R$ of a  new sample with features $\xb\in \R^\ndim $ is given by the inner product $\betab^T\xb$. 
We note that the linear relationship between $y_i$ and $\xb_i$ expressed in \eqref{linear} is in fact quite general. For example, the case where $y_i=f(\xb)+\epsilon_i$, where $f$ is a polynomial function of degree 2, reduces to a linear model by considering the transformed feature space whose features comprise the monomials
$x_{ik}x_{ik'}$, for $1\leq k,k' \leq \ndim$. More generally, the same principle can be
applied to reduce  to  \eqref{linear} any function class spanned by a finite set of basis
functions over $\R^\ndim $ \cite{Hastie09a}.

\subsection{Generalized Least Squares Estimation}

We consider a setup in which the agents choose the precision of the data that they provide. That is, they do not directly provide $y_i$ but rather a perturbed variable $\tilde{y}_i$, which we assume is an unbiased estimate of $y_i$ with variance $\sigma_i^2$. For example, in the case of privacy, agents \emph{distort} their private variable by adding excess noise: each $i\in N$ computes $\tilde{y}_i = y_i+z_i$ where $z_i$ is a zero-mean random variable with variance $\sigma_i^2$; we assume that $\{z_i\}_{i\in N}$ are independent, and are also independent of the inherent noise variables $\{\epsilon_i\}_{i\in N}$. In the case of effort, the variance $\sigma_i^2$ captures the effort exerted by the agent in generating the label $\tilde{y}_i$.   Each agent reveals to the analyst (a) the perturbed variable $\tilde{y}_i$ and (b) the variance $\sigma_i^2$. As a result, the aggregate variance of the reported value is $\sigma^2+\sigma_i^2$ and its precision (the inverse of the aggregate variance) is $\lambda_i\equiv \frac{1}{\sigma^2+\sigma_i^2}$. 

Note that, as a consequence of the above description, our model assumes that the analyst can observe the (true) precision of the private data revealed by the analyst. This is reasonable in settings where the data is stored in a trusted database and the agent grants access to it under a given precision, and the noise is added by a third party (e.g., the database itself). In medical research for instance, one can imagine that the data is stored in a hospital database; a patient would then grant access to it with a given precision and the analyst would receive the perturbed data directly from the hospital. In other applications, such as surveys, one can also imagine that, rather than providing a specific value, agents would provide an interval, whose size naturally translates to  precision.

In turn, having access to the perturbed variables $\tilde{y}_i$, $i\in N$, and the corresponding precisions, the analyst estimates $\betab$ through \emph{generalized least squares} (\gls) estimation. Denote by $\lambdab=[\lambda_i]_{i\in N}$ the vector of precisions and by $\Lambda = \diag(\lambdab)$ the diagonal matrix whose diagonal is given by vector $\lambdab$. Then, the generalized least squares estimator is given by:
\begin{align}
\label{glsq}
\hat{\betab}_{\gls} = \argmin_{\betab\in R^\ndim}\left( \sum_{i\in N} \lambda_i(\tilde{y}_i -\betab^T \xb_i)^2 \right)= (X^T\Lambda X)^{-1}X^T \Lambda\,\tilde{\yb},
\end{align}
where  $\tilde\yb=[\tilde{y}_i]_{i\in N}$ is the $n$-dimensional vector of perturbed variables, and $X=[\xb_i^\T]_{i\in N}\in \R^{n\times \ndim}$ the $n\times \ndim$ matrix whose rows comprise the transposed feature vectors. Throughout our analysis, we assume that $n\ge \ndim$ and that $X$ has rank $d$. 

Note that $\tilde{\yb}\in \R^n$ is a random variable and as such, by \eqref{glsq}, so is $\hat{\betab}_\gls$. It can be shown that $\E(\hat{\betab}_\gls)=\betab$ (i.e., $\hat{\betab}_\gls$ is unbiased), and
$$\cov(\lambdab) \equiv Cov(\hat{\betab}_\gls) = \E\left[(\hat{\betab}_\gls-\betab)^T(\hat{\betab}_\gls-\betab)\right] = (X^\T \Lambda X)^{-1}.$$
The  covariance $V$ captures the uncertainty of the estimation of $\betab$. The matrix 
$$\prcs(\lambdab) \equiv X^T\Lambda X  = \textstyle\sum_{i\in N}\lambda_i \xb_i\xb_i^\T$$ is known as the \emph{precision} matrix. It is positive semidefinite, i.e., $\prcs(\lambdab) \in S_+^\ndim$,
but it may not be invertible: this is the case when $\mathrm{rank}(X^T\Lambda)<d$,  i.e., the vectors $\xb_i$, $i\in N$, for which $\lambda_i>0$, do not span $\R^\ndim$. Put differently, if the set of agents providing useful information does not include $d$ linearly independent vectors, there exists a direction  $\xb\in\R^\ndim$ that is a ``blind spot'' to the analyst: the analyst has no way of predicting the value $\betab^\T\xb$. In this degenerate case the number of solutions to the least squares estimation problem \eqref{glsq} is infinite, and the covariance is not well-defined (it is infinite in all such directions $\xb$).  
Note however that, since $X$ has rank $d$ (and hence $X^T X$ is invertible), the set of $\lambdab$ for which the precision matrix is invertible is non-empty. In particular, it contains $(0, 1/\sigma^2]^n$ since $\prcs(\lambdab)\in S_{++}^d $ if $\lambda_i>0$ for all $i\in N$.

\subsection{Non-Cooperative Game Model of Strategic Data Sources}
\label{sec.game}

The perturbed variables $\tilde{y}_i$ are motivated by the fact that strategic data sources incur a cost to provide high-precision data---for instance, due to privacy concerns, an agent may be reluctant to grant unfettered access to her private variable or release it in the clear. On the other hand, it may be to the agent's advantage that the analyst learns the model $\betab$. In our running medical example, learning that, e.g., a disease is correlated to an agent's weight or her cholesterol level may lead to a cure, which in turn may be beneficial to the agent.

We model the above considerations through cost functions. Recall that the action of each agent $i\in N$ amounts to choosing the noise level of the perturbation, captured by the variance $\sigma_i^2\in [0,\infty]$. For notational convenience, we use the equivalent representation $\lambda_i=1/(\sigma^2+\sigma_i^2)\in [0,1/\sigma^{2}]$ for the action of an agent.
Note that $\lambda_i=0$ (or, equivalently, infinite variance $\sigma_i^2$) corresponds to no participation: in terms of estimation through \eqref{glsq}, it is as if this perturbed value is not reported. 

Each agent $i\in N$ chooses her action $\lambda_i \in [0, 1/\sigma^2]$ to minimize her cost
\begin{align}
 \label{costfun}
J_i(\lambda_i, \lambda_{-i}) = c_i(\lambda_i)+  f(\lambdab),
\end{align}
where we use the standard notation $\lambda_{-i}$ to denote the collection of actions of all agents but $i$. 
The cost function $J_i:\R^{n}_+\to\R_+$ of agent $i\in N$ comprises two non-negative components. 
We refer to the first component $c_i:\R_+\to \R_+$ as the \emph{disclosure cost}: it is the cost that the agent incurs for providing the perturbed variable. The second component  is the \emph{estimation cost}, and we assume that it takes the form $f(\lambdab) = 
F(\cov(\lambdab))$, if $\prcs(\lambdab)\in S_{++}^\ndim$, and $f(\lambdab)=\infty$ otherwise.
The mapping $F:S_{++}^\ndim\to \R_+$ is known as a \emph{scalarization} \cite{boyd2004convex}. It maps the covariance matrix  $\cov(\lambdab)$ 
to a scalar value $F(\cov(\lambdab))$, and captures how well the analyst can estimate the model $\betab$. The estimation cost $f: \R_+^n\to \bar{\R}_+=\R_+\cup\{\infty\}$  is the so-called \emph{extended-value extension} of $F(V(\lambdab))$: it equals $F(V(\lambdab))$ in its domain, and  $+\infty$ outside its domain. 

\paragraph{Main Assumptions.} Throughout our analysis, we make the
following two assumptions:
\begin{assumption} 
\label{privacyassumption} 
The disclosure costs $c_i:\R_+\to \R_+$, $i\in N$, are non-negative, continuous, non-decreasing and convex. 
\end{assumption}
\begin{assumption} 
\label{estimationassumption} 
The scalarization $F:S_{++}^\ndim\to \R_+$ is non-negative, continuous, increasing in the positive semidefinite order, and convex.
\end{assumption}

The monotonicity assumptions in Assumptions~\ref{privacyassumption} and \ref{estimationassumption} are standard and  natural. Increasing the precision $\lambda_i$ leads to a higher disclosure cost. 
In contrast, increasing $\lambda_i$ can only decrease the estimation cost: this is because  decreasing the variance of an agent's provided perturbed variable also decreases the variance in the positive semidefinite sense (as the matrix inverse is a PSD-decreasing function).

The convexity assumption in Assumption~\ref{estimationassumption} is also standard and natural. Intuitively, the naturalness of Assumption~\ref{estimationassumption} stems from the following observation:  the convexity of $F$ implies that the so called \emph{information gain}, i.e., the relative reduction in $F$ as a new label is collected, exhibits a diminishing returns property, as additional labels affect estimation quality less and less.
Scalarizations of positive semidefinite matrices and, in particular, of the covariance matrix $\cov(\lambdab)$, are abundant in statistical inference literature in the context of experimental design~\cite{boyd2004convex,pukelsheim2006optimal,atkinson2007optimum} (also known as batch active learning). Similar to our setting, in experimental design  an analyst  has access to samples with known feature vectors $\xb_i\in \mathbb{R}^d$, $i\in N$, and wishes to conduct a limited number of $k$ experiments, where $k\ll N$, to collect labels $y_i\in \mathbb{R}$ for a subset of these samples. Given  budget $k$, the experimental design problem amounts to determining which labels to collect. The standard approach is to accomplish this by minimizing a scalarization function of the covariance of the estimator applied to the labels selected~\cite{boyd2004convex,pukelsheim2006optimal,atkinson2007optimum}. Three examples of such estimators encountered often in practice are the so-called A-optimality, E-optimality, and D-optimality criteria:
\begin{align}\label{specific}
	F_1(\cov) &= \trace (\cov), 
	& F_2(\cov) &= \| \cov \|_F^2,
	& F_3(\cov) &= \log\det (I+\cov),
\end{align}
where $\| \cdot \|_F$ is the Frobenius norm and $I$ is the identity matrix.
All three scalarizations satisfy Assumption~\ref{estimationassumption}. The convexity of these scalarizations implies that, if repetitions are allowed (i.e., an experiment can be conducted multiple times), the analyst can determine which fraction of her experiments should be performed on a given sample by solving a convex optimization problem (see, e.g., \cite{boyd2004convex}). On the other hand, if repetitions are not allowed, convexity implies that experimental design can be cast as a submodular maximization problem subject to cardinality constraints (see, e.g., \cite{horel2013budget}), which is NP-hard for the  objectives in \eqref{specific} but admits a poly-time approximation. Submodular maximization arises precisely due to the aforementioned diminishing returns property of the information gain under new labels; this, in turn, a direct consequence of Assumption~\ref{estimationassumption}.

Note that, as a further consequence of Assumption \ref{estimationassumption}, the extended-value extension $f$ is convex (in $\lambdab$). The convexity of $F(V(\cdot))$ follows from the fact that it is the composition of the increasing convex function $F(\cdot)$ with the matrix convex function $V(\cdot)$; the latter is convex because the matrix inverse is matrix convex and the precision matrix $A(\lambdab)$ is an affine function of $\lambdab$.

\paragraph{Additional Assumptions.}
Our result on Nash equilibrium existence and uniqueness (Theorem~\ref{thm.Nash}) relies on Assumptions~1 and 2. Our bounds on the price of stability (Theorem~\ref{generalizedBound}) and our Aitken-type results (Theorem~\ref{GLSoptimaltofactor}) use two additional assumptions that further constrain the shape of the disclosure costs and the scalarization function:
\begin{assumption}\label{homogeneityc}
There exist $1\le\pone\le\ptwo\in\R_+\cup\{+\infty\}$ such that, for all $i \in N$, the disclosure cost $c_i : \R_+ \rightarrow \R_+$ satisfies:
\begin{align}
    a^{\pone}c_i(\lambda) \le c_i(a\lambda) \le a^{\ptwo}c_i(\lambda),   \qquad     \text{for all $\lambda \in \R_+$ and $a\ge1$}. 
\end{align}
\end{assumption}

\begin{assumption}\label{homogeneityf}
There exists $q\ge 1$ such that the scalarization $F:S_{++}^\ndim\to \R_+$ is $q$-homogeneous, i.e., it satisfies:
\begin{align}
  F(aM) = a^{q}F(M), \qquad \text{for all $M \in S_{++}^\ndim$ and $a \ge 1$}.
\end{align}
\end{assumption}

Intuitively, Assumption~\ref{homogeneityc} captures ``near-homogeneity'' of the disclosure cost functions. It is, for example, satisfied when all agents have monomial disclosure costs $c_i(\lambda)=r_i\lambda^{p_i}$, where $r_i$ is a constant, with different exponents $p_i\in[\pone,\ptwo]$. Assumption~\ref{homogeneityf} is also a homogeneity assumption. It holds for a broad class of interesting scalarizations, such as any norm taken to any power. In particular, it holds for $F_1$ and $F_2$ in \eqref{specific}, i.e., the trace and the squared Frobenius norm (with $q = 1$ and $q = 2$ respectively), which are classical scalarizations in the statistical inference literature in the context of experimental design~\cite{boyd2004convex,pukelsheim2006optimal,atkinson2007optimum}. Note that  Assumption~\ref{homogeneityf} also implies that $f(a\lambdab) = a^{-q}f(\lambdab)$, for all $\lambdab \in [0, 1/\sigma^{2}]^n$ and $a \ge 1$.

\paragraph{Game notation}
We denote by $\Gamma = \langle N,  [0, 1/\sigma^{2}]^n, \left( J_i \right)_{i\in N} \rangle$ the game with set of agents $N = \{1, \cdots, n\}$, where each each agent $i\in N$ chooses her action $\lambda_i$ in her action set $[0, 1/\sigma^{2}]$ to minimize her cost $J_i: [0, 1/\sigma^{2}]^n\to \R_+$, given by \eqref{costfun}. We refer to a $\lambdab\in  [0, 1/\sigma^{2}]^n$ as a \emph{strategy profile} of the game $\Gamma$.
We analyze the game as a \emph{complete information game}, i.e., we assume that the set of agents, the action sets and utilities are known by all agents.

\section{Nash Equilibria}\label{nash}
%!TEX root = techreport.tex

We begin our analysis by characterizing the Nash equilibria of the game $\Gamma$. In the game $\Gamma$, each agent chooses her contribution $\lambda_i$ to minimize her cost. A Nash equilibrium (in pure strategy) is a strategy profile $\lambdabeq$ satisfying  
\begin{equation*}
\lambdaeq_i \in \argmin_{\lambda_i} J_i(\lambda_i, \lambdaeq_{-i}), \quad \textrm{ for all } i\in N.
\end{equation*}

Observe first that $\Gamma$ is a potential game \cite{Monderer1996a}. Indeed, define the function $\Phi: [0, 1/\sigma^{2}]^n\to \bar{\R}$ such that  
\begin{equation}
\label{eq.defPhi}
\Phi(\lambdab) = f(\lambdab) + \sum_{i\in N} c_i(\lambda_i), \quad (\lambdab \in [0, 1/\sigma^{2}]^n). 
\end{equation}
Then for every $i\in N$ and for every $\lambda_{-i} \in [0, 1/\sigma^{2}]^{n-1}$, we have
\begin{equation}
\label{eq.potential}
J_i(\lambda_i, \lambda_{-i}) - J_i(\lambda^{\prime}_i, \lambda_{-i}) = \Phi(\lambda_i, \lambda_{-i}) - \Phi(\lambda^{\prime}_i, \lambda_{-i}), \quad \forall \lambda_i, \lambda^{\prime}_i \in [0, 1/\sigma^{2}].
\end{equation}
Therefore, $\Gamma$ is a potential game with potential function $\Phi$. 
From \eqref{eq.potential}, we see that (as for any convex potential game) the set of Nash equilibria coincides with the set of local minima of function $\Phi$.

Note that there may exist Nash equilibria $\lambdabeq$ for which $f(\lambdabeq)=\infty$. For instance, if $\ndim\ge2$, $\lambdabeq = 0$ is a Nash equilibrium. Indeed, in that case, no agent has an incentive to deviate since a single $\lambda_i>0$ still yields a non-invertible precision matrix $A(\lambdab)$.  In fact, any profile $\lambdab$ for which $A(\lambdab)$ is non-invertible, and remains so under unilateral deviations, is an equilibrium. 
We call such Nash equilibria (at which the estimation cost is infinite) \emph{trivial}. Existence of trivial equilibria can be avoided using slight model adjustments: for instance, one can alter the game definition to disallow infinite variances. Alternatively,  the existence of $d$ non-strategic agents whose feature vectors span $\R^d$ is also sufficient to enforce a finite covariance at all $\lambdab$ across strategic agents.

In the remainder, we focus on the more interesting \emph{non-trivial} equilibria. 
Using the potential game structure of $\Gamma$, we derive the following result. 
\begin{theorem}
\label{thm.Nash}
Under Assumptions \ref{privacyassumption} and \ref{estimationassumption}, there exists a unique non-trivial equilibrium of the game $\Gamma$.
\end{theorem}

This result is proved in Appendix \ref{proof.thm.Nash}.
The potential game structure of $\Gamma$ has another interesting implication: if agents start from an initial strategy profile $\lambdab$ such that $f(\lambdab) < \infty$, the so called \emph{best-response dynamics} converge towards the unique non-trivial equilibrium (see, e.g., \cite{Sandholm10a}). This implies that the non-trivial equilibrium is the only equilibrium reached when, e.g., all agents start with non-infinite noise variance.

\section{Price of Stability}\label{stability}
%!TEX root = techreport.tex

Having established the uniqueness of a non-trivial equilibrium in our game, we turn our attention to issues of efficiency. We define the \emph{social cost} function $\SC:\R^n\to \R_+$ as the sum of all agent costs, and say that a strategy profile $\lambdab^\opt$ is \emph{socially optimal} if it minimizes the social cost,  i.e.,
\begin{equation*}
\SC(\lambdab) = \sum_{i\in N} c_i(\lambda_i)+nf(\lambdab), \qquad\text{and}\qquad \lambdab^\opt \in \argmin_{\lambdab\in [0,1/\sigma^{2}]^n} \SC(\lambdab). 
\end{equation*}
Let $\opt = \SC(\lambdab^\opt)$ be the minimal social cost.
We define the \emph{price of stability} (\emph{price of anarchy}) as the ratio of the social cost of the best (worst) Nash equilibrium in $\Gamma$ to \opt, i.e.,
$$\pos = \min_{\lambdab \in \NE} \frac{ \SC(\lambdab)}{\opt} , \qquad\text{and}\qquad \poa = \max_{\lambdab \in \NE} \frac{ \SC(\lambdab)}{\opt} , $$
where $\NE\subseteq [0,1/\sigma^{2}]^n$ is the set of Nash equilibria of $\Gamma$. 
Clearly, in the presence of trivial equilibria, the price of anarchy is infinity. We thus turn our attention to determining the price of stability. Note however that since the non-trivial equilibrium is unique (Theorem~\ref{thm.Nash}), the price of stability and the price of anarchy coincide under the slight model adjustments discussed in Section~\ref{nash} that eliminate trivial equilibria. 

The fact that our game admits a potential function has the following immediate consequence (see, e.g., \cite{Schafer11a,Sandholm10a}):
\begin{theorem}\label{general}
Under Assumptions~\ref{privacyassumption} and~\ref{estimationassumption}, $\pos\leq n$.
\end{theorem}
 
The proof of Theorem~\ref{general} can be found in Appendix~\ref{proof.general}. Improved bounds can be obtained for specific estimation and disclosure cost functions. The following result provides tighter bounds when the disclosure costs and the scalarization satisfy Assumptions~\ref{homogeneityc} and \ref{homogeneityf}.
\begin{theorem}\label{generalizedBound}
In addition to Assumptions \ref{privacyassumption} and \ref{estimationassumption}, assume that the disclosure cost functions satisfy Assumption~\ref{homogeneityc} with $\pone\ge1$ and $\ptwo\in\R\cup\{\infty\}$ and that the scalarization F satisfies Assumption~\ref{homogeneityf} for some $q \ge 1$. Then, the price of stability satisfies $\pos \leq {n}^{\frac{q}{\pone + q}}$.  Additionally, for all $\pone, q \ge 1$, and all $\varepsilon > 0$, there exists a game in which the estimation cost and the disclosure costs satisfy Assumptions~\ref{homogeneityc} and ~\ref{homogeneityf}, respectively,  such that $\pos \geq {n}^{\frac{q}{\pone+q}}(1-\varepsilon)$.
\end{theorem}

The proof of Theorem~\ref{generalizedBound} can be found in
Appendix~\ref{proofofthmgeneralized}. Note that, as the bound does not
depend on $\ptwo$, we can set $\ptwo=\infty$ in
Assumption~\ref{homogeneityc}, which is equivalent to replacing this assumption by
\begin{align*}
  a^{\pone}c_i(\lambda) \le c_i(a\lambda),\qquad
   \text{for all $\lambda \in \R_+$ and $a \ge 1$.}
\end{align*}
The proof of the upper bound relies on deriving a ``good'' solution from the social optimum and showing that, if the $\pos$ is too high, this ``good'' solution attains a lower potential than a Nash equilibrium (a contradiction).
The proof of the lower bound in Theorem~\ref{generalizedBound} relies on explicitly characterizing  the socially optimal profile in a certain game class, and showing it equals the Nash equilibrium $\lambdabeq$ multiplied by a scalar. We note that the theorem states that, among monomial disclosure costs and for any estimation cost satisfying Assumption \ref{homogeneityf}, the largest $\pos$ is $n^{\frac{q}{1 + q}}$ and is attained for linear disclosure costs. Similarly, among all estimation costs satisfying Assumption \ref{homogeneityf} and  all disclosure costs satisfying the assumptions presented in Theorem~\ref{generalizedBound}, the largest $\pos$ is $n$; this is approached as $q$ tends to infinity. 
 
We note that a similar worst-case efficiency of linear functions among
convex cost families has also been observed in the context of other
games, including routing \cite{Roughgarden02a} and resource allocation
games \cite{Johari04a}. As such, Theorem~\ref{generalizedBound}
indicates that this behavior emerges in our linear regression game as
well but only concerning the disclosure cost: We observe a worst-case
efficiency of linear functions in this game for the disclosure cost but a
worst-case efficiency of highly convex functions for the estimation
cost.

\section{An Aitken-Type Theorem for Nash Equilibria}\label{gauss}
%!TEX root = techreport.tex

Until this point, we have assumed that the analyst uses the generalized least-square estimator~\eqref{glsq} to estimate model $\betab$. In the non-strategic case,  where  $\lambdab$ (and, equivalently, the added noise variance) is fixed, the generalized least-square estimator is known to satisfy a strong optimality property: the so-called Aitken/Gauss-Markov theorem, which we briefly review below, states that it is a ``Best Linear Unbiased Estimator'', a property commonly referred to as BLUE. In this section, we investigate how this result extends in the strategic case, i.e., when  $\lambdabeq$ is not a priori fixed but is the equilibrium reached by agents; crucially, the latter depends on the estimator used by the analyst.

\subsection{Linear Unbiased Estimators and the Aitken Theorem}

A \emph{linear} estimator $\hat{\betab}_L$ of the model $\betab$ is a linear map of the perturbed variables $\tilde{\yb}$; i.e., it is an estimator that can be written as $\hat{\betab}_L = L\tilde{\yb}$ for some matrix $L\in \R^{\ndim \times n}$. A linear estimator is called \emph{unbiased} if $\E[L\tilde{\yb}]=\betab$ (the expectation taken over the inherent and extra noise). Recall by~\eqref{glsq} that the generalized least-square estimator $\hat\betab_{\gls}$ is an unbiased linear estimator with $L=(X^T\Lambda X)^{-1}X^T\Lambda$ and covariance $Cov(\hat{\betab}_\gls) = (X^\T \Lambda X)^{-1}$.
Any linear estimator $\hat{\betab}_L=L\tilde{\yb}$ can be written
without loss of generality as
\begin{align} L = (X^T\Lambda X)^{-1}X^T \Lambda + D^T,\label{LE}\end{align} where
	\begin{align}\label{D}
		D =D(X)\in \mathbb{R}^{d\times n},
	\end{align}%
is a matrix that may depend on $X$ but does not depend on $\Lambda$.
It is easy to verify that $\hat\betab_L$ is unbiased if and only if 
\begin{equation}
\label{Dconst}
D^T X = 0;
\end{equation} 
in turn, using this result, the covariance of any linear unbiased estimator can be shown to be  
\begin{equation}
  \label{eq:cov}
Cov(\hat{\betab}_L) =  (X^\T \Lambda X)^{-1} +D^\T  \Lambda^{-1} D  \succeq  Cov(\hat{\betab}_\gls). 
\end{equation}

In other words, the covariance of the generalized least-square estimator is minimal in the positive-semidefinite order among the covariances of \emph{all linear unbiased estimators}. This optimality result is known as the Aitken theorem~\cite{aitken}. Applied specifically to homoscedastic noise (i.e., when all noise variances are identical), it is known as the Gauss-Markov theorem~\cite{Hastie09a}, which establishes the optimality of the ordinary least squares estimator. Both theorems provide a strong argument in favor of using least squares to estimate $\betab$, in the presence of fixed noise variance (i.e., non-strategic agents).

\subsection{Extension of the Non-Cooperative Game to Linear Unbiased
  Estimators}

Suppose now that the data analyst uses a linear unbiased estimator $\hat\betab_L$ with a given matrix $L\in \R^{n\times\ndim}$ which may depend on $X$. Similarly to the model introduced in Section~\ref{sec.game}, we define a game $\Gamma_{L}$ in which each agent $i$ chooses her $\lambda_i$ to minimize her cost; this time, however, the estimation cost depends on the variance of $\hat\betab_L$. A natural question to ask is the following: it is possible that, despite the fact that the analyst uses an estimator that is ``inferior'' to $\hat\betab_\gls$ in the BLUE sense, an equilibrium reached under $\hat\betab_L$ is \emph{better} than the equilibrium reached under $\hat\betab_\gls$ in terms of equilibrium estimation cost? If so, despite the Aitken theorem, the data analyst would have an incentive to use $\hat\betab_L$ instead and to inform the agents that she will use $\hat\betab_L$ and not $\hat\betab_{\gls}$.

\sloppy
In this section, we provide both a positive and a negative answer to this question, depending on specific assumptions on the disclosure costs. Formally, we consider the game $\Gamma_{L} = \langle N, [0, 1/\sigma^2]^n, (J_i)_{i\in N} \rangle$  defined as in Section~\ref{sec.game}, except that the estimation cost is the extended-value extension of $F(V_L(\lambdab))$ with  
\begin{align}
  V_L(\lambdab) \equiv (X^\T \Lambda X)^{-1} +D^\T  \Lambda^{-1} D,
  \quad\quad (\Lambda = \diag \lambdab), \label{newV}
\end{align}
where $D$ is defined as in \eqref{Dconst}. 

Recall that $D$ may depend on $X$ but does not depend on the
precision $\Lambda$.
\fussy
Then, $\Gamma_L$ is still a potential game with potential function $\Phi(\lambdab) = f_L(\lambdab) + \sum_{i\in N} c_i(\lambda_i)$ where $f_L(\lambdab)=F(V_L(\lambdab))$. This potential function has the same form as the potential of the original game, given by~\eqref{eq.defPhi}.  Moreover, as the function $V_L(\cdot)$ given by \eqref{newV} is a matrix convex function, the extended-value extension $f_L(\cdot)$ is convex. This shows that the potential function is convex.  Since the proof of Theorem~\ref{thm.Nash} mostly relies on the convexity of the potential, a straightforward adaptation yields the following result.
\begin{theorem}
Under Assumptions \ref{privacyassumption} and \ref{estimationassumption}, for any linear estimator $L\in \R^{n\times \ndim}$, there exists a unique non-trivial equilibrium of the game $\Gamma_L$.
\end{theorem}
As for the case of $\gls$, this result follows from the uniqueness of a minimizer of the potential function attained in the effective domain. In what follows, we denote the unique non-trivial equilibrium of $\Gamma_L$ by $\lambdabeq_{L}$ and we denote by $\lambdabeq_{\gls}$ the equilibrium of the game $\Gamma$ with the same parameters except for the estimator.

\subsection{Optimality of \gls}

\subsubsection{Theoretical Bound and Optimality Condition}

For a given linear unbiased estimator $L$, the estimation cost at equilibrium is $f_{L}(\lambdabeq_{L})$. We say that a linear estimator is efficient if it provides a small estimation cost at equilibrium. In the following theorem, we provide both a negative and a positive result about the efficiency of $\gls$: on the one hand, $\gls$ is not always the most efficient estimator; on the other hand, under Assumptions~\ref{homogeneityc} and \ref{homogeneityf}, the ratio between the estimation cost at equilibrium of $\gls$ and any other estimator is bounded by $\frac{\ptwo (q + \pone)}{\pone (q + \ptwo)}$; in order words, $\gls$ is never too far from the most efficient estimator.

\begin{theorem}\label{GLSoptimaltofactor}
Assume that the disclosure cost and scalarization functions satisfy Assumptions \ref{privacyassumption} and \ref{estimationassumption}. Then:

\noindent (i) There exists a game $\Gamma$ such that $\gls$ is not the most efficient estimator; i.e., there exists an unbiased linear estimator $L$ such that, for these game parameters,
  \begin{align*}
    f_L(\lambdabeq_{L}) < f_{\gls}(\lambdabeq_{\gls}). 
  \end{align*}

 \noindent (ii) For all games that additionally satisfy Assumptions \ref{homogeneityc} and
  \ref{homogeneityf}, $\gls$ is $\frac{\ptwo (q + \pone)}{\pone (q + \ptwo)}$-optimal,
  i.e., for all unbiased estimators $L$
  \begin{equation*}
     f_{\gls}(\lambdabeq_{\gls}) \le  \frac{\ptwo (q + \pone)}{\pone (q + \ptwo)} f_{L} (\lambdabeq_{L}).
  \end{equation*}
\end{theorem}

The proof is provided in Appendix~\ref{GLSOptimal}. 
Note that the bound in Theorem~\ref{GLSoptimaltofactor}\emph{(ii)} is clearly smaller than or equal to $\ptwo/\pone$. By remarking that it can be written as $(1+\frac{q}{\pone})/(1+\frac{q}{\ptwo})$, it is easy to see that it is also smaller than or equal to $1+q$. This shows that $\gls$ is $\ptwo/\pone$-optimal for any $q$ and $(1+q)$-optimal for any $\pone, \ptwo$. 
Note also that Theorem~\ref{GLSoptimaltofactor}\emph{(ii)} trivially implies the following:
\begin{corollary}\label{Optimality_of_GLS}
  Under Assumptions \ref{privacyassumption} and \ref{estimationassumption}, Assumption \ref{homogeneityc} with $\pone = \ptwo = p$, and Assumption \ref{homogeneityf}, $\gls$ is the most efficient estimator. 
\end{corollary}
Note that $\pone = \ptwo = p$, which literally translates to $c_i(a\lambda) = a^p c_i(\lambda)$ for all $i\in N$, $\lambda \in \R_+$ and $a\ge1$, means that all agents have monomial costs functions with the same exponent. Put differently, for all $i$, there exists a constant $r_i > 0$ such that $c_i(\lambda_i) = r_i \lambda_i^p$.
Theorem~\ref{GLSoptimaltofactor}\emph{(i)} may seem counter-intuitive
as $\gls$ is optimal in the case of non-strategic agents: by Aitken's
theorem, if  precisions are fixed and known, then the best linear
unbiased estimator is $\gls$, i.e., for all $\lambdab$:
$f_L(\lambdab)>f_{\gls}(\lambdab)$. Our result demonstrates that this
is not the case with strategic agents. 

\subsubsection{Numerical Illustration of the Non-Optimality of \gls} 
    
    The proof of the non-optimality of \gls 
  (Theorem~\ref{GLSoptimaltofactor}(i)) is a constructive proof that
  uses a counter-example with two agents in a one-dimensional model
  ($d=1$) where both agents have the same public data. This raises the question of whether the suboptimality of \gls arises in higher dimensions or, more generally, in more complicated scenarios. 
  Although extending our analytical proof to more general cases appears to be difficult,
  in this section, we provide three numerical counter-examples
  that illustrate the gap of sub-optimality of \gls. In particular, our numerical counter-examples suggest that the sub-optimality of \gls is not limited to the simple counter-example of our analytical proof.

  These counter-examples are constructed by using an estimator
  $L(\delta)$ equal to $\gls$ plus a small perturbation
  term of the form $\delta$ times $D^T$, i.e., 
  $$L(\delta)= \gls+ \delta D^T \equiv (X^T\Lambda X)^{-1}X^T \Lambda+\delta D^T,$$
  for an appropriately selected $D$. The idea behind our
  counter-examples is that when using a perturbed estimator (with perturbation $\delta>0$), that
  is less accurate than $\gls$ under non-strategic agents, some agents
  will tend to choose a higher precision than under $\gls$ at
  equilibrium.  In all of our numerical examples, a small enough
  $\delta$ leads to an estimation cost at equilibrium smaller than the
  one of $\gls$ because some agents will use a higher precision. When
  $\delta$ increases too much, the gain brought by the higher
  precision of agents is  canceled by the loss of precision
  that is caused by using the estimator $L(\delta)$ that is less
  precise than $\gls$. 

  In all of our examples, the equilibrium costs
  of the estimators are very close to that of $\gls$ and our examples are
  far from attaining the bound
  $\frac{\ptwo (q + \pone)}{\pone (q + \ptwo)}$ provided by
  Theorem~\ref{GLSoptimaltofactor}. We believe that this bound is
  loose and can probably be refined.

  We present three examples because each example is of independent interest.  The first two involve $1$-dimensional models ($d=1$). In the first
  example, we use a perturbation term that affects all agents. For
  this example, we believe that $\gls$ is sub-optimal only when the
  two exponents $\pone$ and $\ptwo$ are significantly different.  In
  the second example, we use a perturbation that only affects  two
  ``less generous'' agents. This allows us to build a counter-examples
  with similar disclosure costs (with exponents $\pone=1.01$ and
  $\ptwo=1.1$). Our third example includes several counter examples in
  settings for different values  $d\ge 2$. This setting has $d$ symmetrical agents  
 and a single $(d+1)$-th agent whose public vector
  $\xb_{d+1}$ is significantly different.
  
  To ensure the reproducibility of these results, we make the code used to compute
  the equilibria and to produce the figures in this section publicly available.\footnote{\url{https://github.com/ngast/strategicLinearRegression}.}

\paragraph{Example 1 (1-dimensional model with two agents).} We consider a
$1$-dimensional model ($d=1$) with two agents ($n=2$) in which the
public data of each agent is $x_i=1$. For such a game, the estimator
$\gls$ is
$(X^T\Lambda X)^{-1}X^T\Lambda \tilde{\yb}=
(\lambda_1+\lambda_2)^{-1}\lambdab^T \tilde{\yb} $ and its covariance
is $1/(\lambda_1+\lambda_2)$. We consider a linear estimator
$L(\delta)$ of the form
\begin{equation*}
\gls+\left[\begin{smallmatrix}\sqrt{\delta}\\-\sqrt{\delta}\end{smallmatrix}\right]=\left[\begin{smallmatrix}\lambda_1/(\lambda_1+\lambda_2)+\sqrt{\delta}\\ \lambda_2/(\lambda_1+\lambda_2)-\sqrt{\delta}\end{smallmatrix}\right].
\end{equation*}
According to \eqref{eq:cov}, its covariance is
$1/(\lambda_1+\lambda_2) + \delta/\lambda_1+\delta/\lambda_2$, where
$\delta/\lambda_1+\delta/\lambda_2$ is the loss of precision due to
using a linear estimator that is less precise than $\gls$.  We assume
that the disclosure cost of Agent~1 is
$c_1(\lambda) = \lambda^{1.01}$ ($\pone=1.01$) while the disclosure
cost of Agent~2 is $c_2(\lambda) = \lambda^{20}$
($\ptwo=20$). The scalarization function is the identity, which means
that
$f_{L(\delta)}(\lambdab)=1/(\lambda_1+\lambda_2) +
\delta/\lambda_1+\delta/\lambda_2$. We set the maximal precision to $1/\sigma^2=1$.

In Figure~\ref{fig:GLS_ex1}(a), we plot the estimation cost at
equilibrium $f_{L(\delta)}(\lambdabeq_{L(\delta)})$ as a function of
$\delta$. We observe that with \gls we get an estimation cost of
approximately $0.99$. When $\delta$ increases, the estimation cost at
equilibrium decreases up to $\delta = 0.012$ for which it reaches
approximately $0.96$. This decrease is explained by the fact that for
small $\delta$, the gain due to a higher precision used by Agent~1 is
larger than the loss of precision $\delta/\lambda_1+\delta/\lambda_2$.
When $\delta$ exceeds $0.012$, this loss of precision is more
important than the gain due to higher precision. This behavior is
further illustrated in Figure~\ref{fig:GLS_ex1}(b), where we plot the
precision released by the two agents. We observe that the precision
of Agent~1 increases with $\delta$ while the precision of Agent~2
decreases (slightly).

\begin{figure}[ht]
  \centering
  \begin{tabular}{cc}
    \includegraphics[width=.48\linewidth]{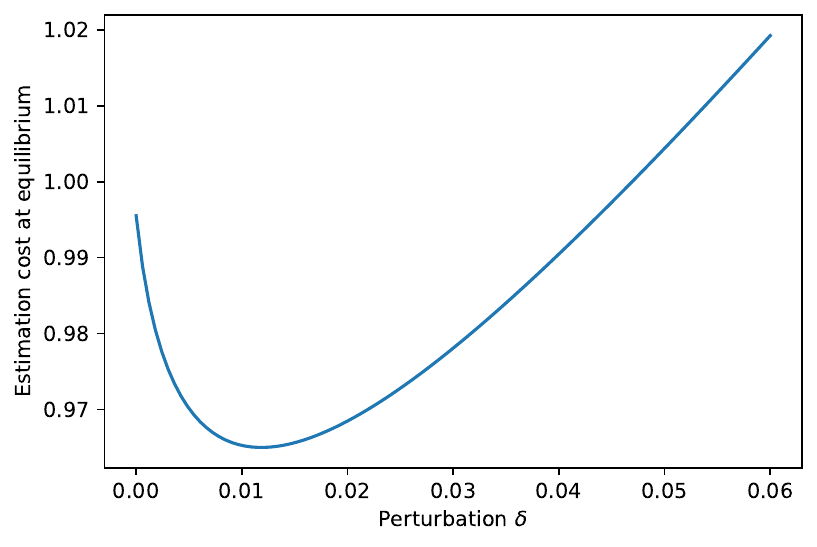}
    &\includegraphics[width=.48\linewidth]{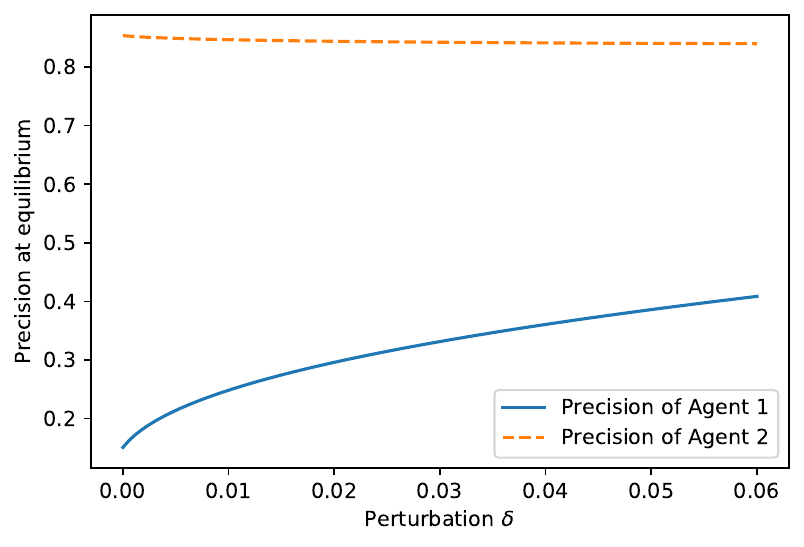}\\
    (a) Estimation cost 
    & (b) Precision of agents 
  \end{tabular}
  \caption{Counter-example 1: Estimation cost and precision of agents
    as a function of the perturbation $\delta$. }
  \label{fig:GLS_ex1}
\end{figure}

\paragraph{Example 2 (1-dimensional model with four agents).} We consider a
$1$-dimensional game with four agents in which the public data of each
agent equals $x_i = 1$. Agents
1 and 2 have disclosure costs $c_i(\lambda) = \lambda^{1.01}$ while
Agents 3 and 4 have disclosure costs
$c_i(\lambda) = \lambda^{1.1}$. We consider a linear unbiased
estimator that is equal to $\gls$ plus a perturbation cost that only
affects the first two agents:
$D=[\sqrt{\delta},-\sqrt{\delta},0,0]$. Note that this perturbation is
only applied to the most selfish agents as they are the ones we must
incentivize to give more. We set the maximal precision to $1/\sigma^2=1$.

In Figure~\ref{fig:GLS_ex2}, we plot the
estimation cost at equilibrium $f_{L(\delta)}(\lambdabeq_{L(\delta)})$
as a function of $\delta$. With $\gls$ ($\delta=0$), we get an
estimation cost of $0.9955$, which is larger than the value $0.9950$
that we obtain for $\delta=3.10^{-4}$. As for Example~1, when $\delta$
increases, the precisions used by the least
generous agents (Agents~1 and 2) increase while the
precisions of the most generous agents
decrease.

\begin{figure}[ht]
  \centering
  \begin{tabular}{cc}
    \includegraphics[width=.48\linewidth]{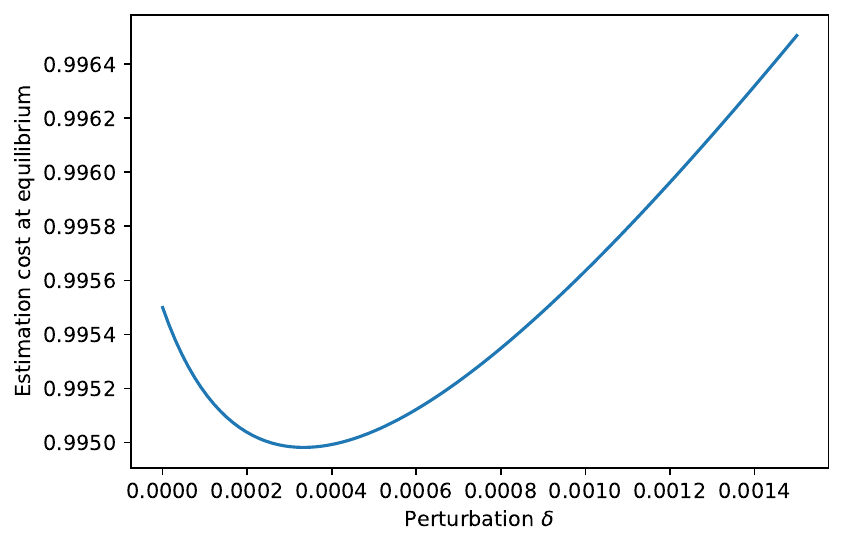}
    &\includegraphics[width=.48\linewidth]{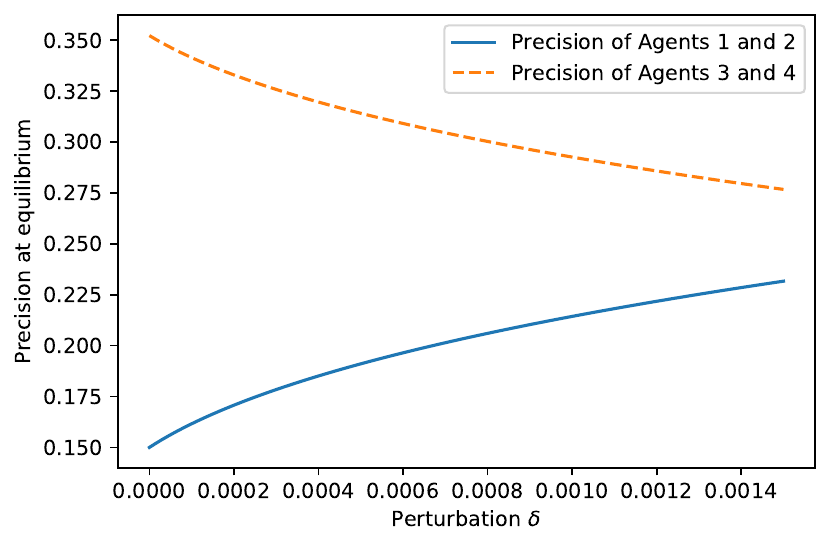}\\
    (a) Estimation cost 
    & (b) Precision of agents
  \end{tabular}
  \caption{Counter-example 2: Estimation cost and precision of agents
    as a function of the perturbation $\delta$. }
  \label{fig:GLS_ex2}
\end{figure}

While the previous two counter-examples are in dimension 1 and with agents that all have $x_i = 1$, the sub-optimality of \gls is not limited to that case. To illustrate that, we consider in the next counter-example models in dimension $d$ with $d\ge 2$. Note that, as we assume that matrix $X$ has rank $d$, we need at least $d$ players whose feature vectors $\xb_i$'s span the $d$ dimensions. Note also that, with $d$ players in $d$ dimensions, $\gls$ is the only linear unbiased estimator. Indeed, as matrix $X$ would then be invertible, the condition in \eqref{Dconst} leads to $D^T = 0$. In Example 3 below, we consider the simplest case of models with $d+1$ agents, though it is clear that one could construct similar counter examples with any number of agents larger than or equal to $d+1$.

  \paragraph{Example 3 ($d$-dimensional models with $d+1$ agents).} We
  consider a $d$-dimensional game with $d+1$ agents. The public data
  of the first $d$ agents spans the $d$ dimensions: $\xb_i$ is
  a vector where all components equal $0$ except the $i$th one that
  is equal to $1$. All components of the public data of Agent $d+1$
  are equal to $1/d$: $\xb_{d+1}=[1/d, \cdots, 1/d]^T$.  We assume
  that the disclosure costs of the first $d$ agents are
  $c_i(\lambda)=\lambda^{20}$ (for $i\in\{1, \cdots, d\}$), and the
  disclosure cost of the last agent is
  $c_{d+1}(\lambda)=\lambda^{1.5}$.  We set the maximal precision to
  $1/\sigma^2=1$.

  The perturbation matrix $D$ is a $(d+1)\times d$ matrix whose first
  column is $\sqrt{\delta}[1, \cdots, 1,-d]$, all other entries being $0$. Hence,
  the public feature matrix $X$ and the perturbation matrix $D$ are
  the following $(d+1)\times d$ matrices:
  \begin{align}
    \label{eq:Xpublic}
    X = \left[
    \begin{array}{cccc}
      1&&0\\
       &\ddots&\\
       0&&1\\
      1/d&\dots&1/d
    \end{array}
                 \right],
              &&D = \left[
                \begin{array}{cccc}
                  \sqrt{\delta}&0&0&\dots\\
                  \vdots&0&0&\dots\\
                  \sqrt{\delta}&0&0&\dots\\
                  -d\sqrt{\delta}&0&0&\dots
                \end{array}
                                      \right].
  \end{align}
  It is easy to verify that $D^TX=0$, which implies $L(\delta)=\gls+D^T$
  is an unbiased estimator.

In Figure~\ref{fig:higher-dim}, we report the estimation cost at
equilibrium $f_{L(\delta)}(\lambdabeq_{L(\delta)})$ as a function of
$\delta$. We consider models of dimension $d\in\{2,5,10,15\}$. We
observe that for all dimensions $d$, the behavior is similar to the
one observed in Figure~\ref{fig:GLS_ex1}(a) and \ref{fig:GLS_ex2}(a):
when $\delta$ is small enough, using the estimator $L(\delta)$
provides a higher precision at equilibrium (i.e., a lower equilibrium estimation cost as seen on the graphs). This comes from the fact
that when $\delta$ increases, the precision at equilibrium provided by
Agent~$d+1$ increases with $\delta$ whereas the precision provided by
Agents $1$ to $d$ is almost independent of $\delta$. When $\delta$
increases too much, the estimation cost increases again because of
the non-optimality of the estimator $L(\delta)$ (for given individual precisions). We also observe that
the maximal gain that can be obtained by using an estimator other than
$\gls$ (and the perturbation $\delta$ for which it is achieved with our particular perturbation matrix $D$) seems to decrease when the dimension $d$ increases.

\begin{figure}[ht]
  \centering
  \begin{tabular}{cc}
    \includegraphics[width=.45\linewidth]{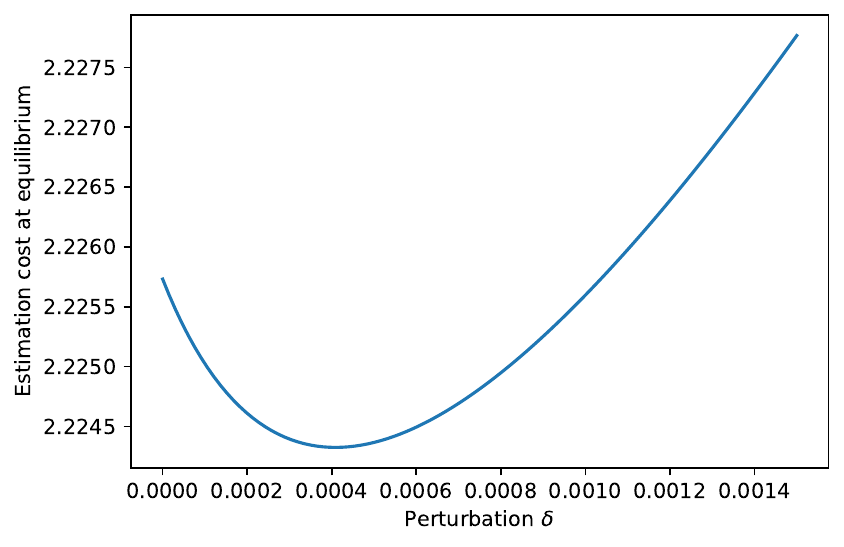}
    &\includegraphics[width=.45\linewidth]{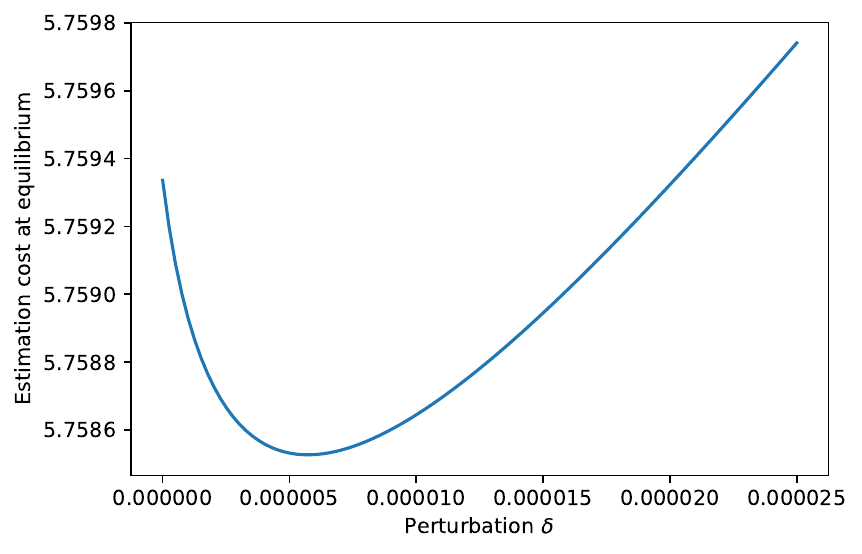}\\
    (a) $d=2$ & (b) $d=5$\\
    \includegraphics[width=.45\linewidth]{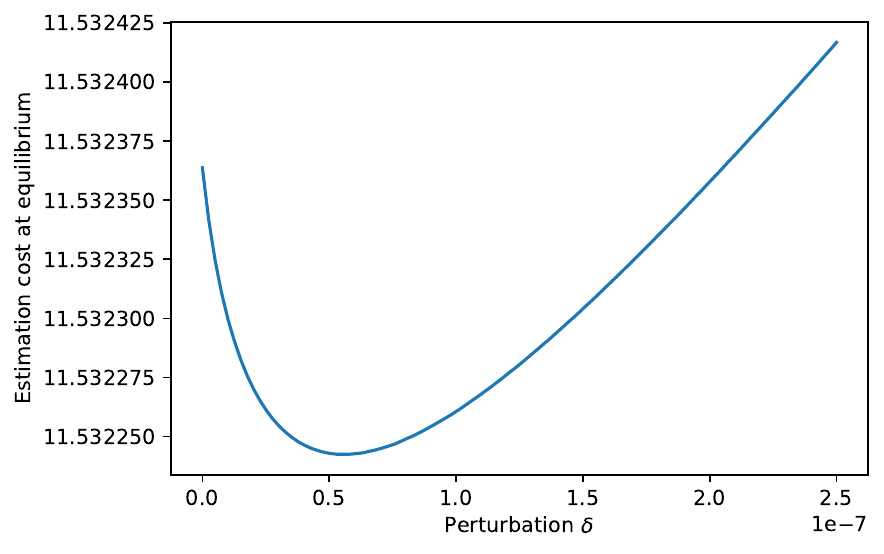}
    &\includegraphics[width=.45\linewidth]{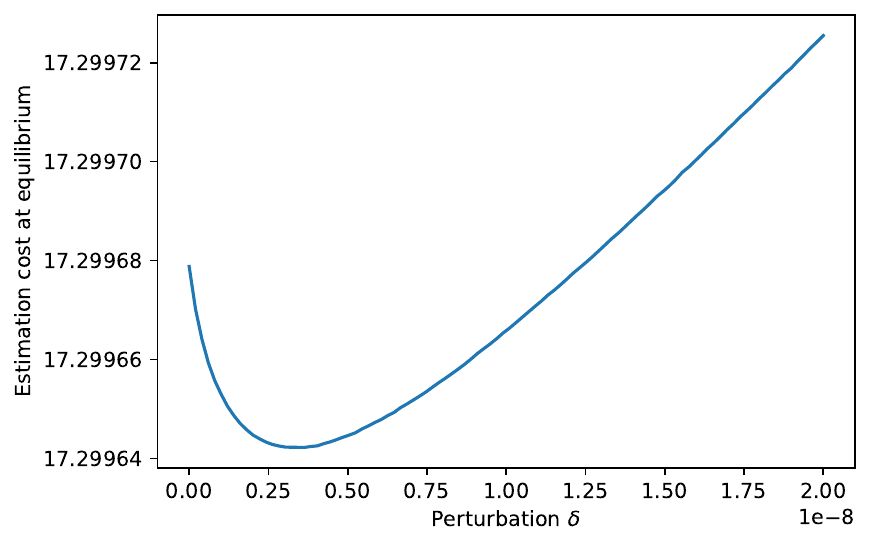}\\
    (c) $d=10$ & (d) $d=15$
  \end{tabular}
  \caption{Counter-example 3: Estimation cost and precision of agents
    as a function of the perturbation $\delta$ for models in dimension $d \ge 2$. }
  \label{fig:higher-dim}
\end{figure}

Finally, although the public feature matrix $X$ in \eqref{eq:Xpublic} has a particular form, many $d$-dimensional
models with $d+1$ agents can be cast in this model via an appropriate change of basis. In fact, we
conjecture that for any matrix of public features $X$ with at least $d+1$
agents, there exist disclosure costs such that \gls is not optimal.

\section{Concluding Remarks}\label{conclusions}
%!TEX root = techreport.tex

This paper studies linear regression in the presence of strategic data sources, modeling the precision choice as a non-cooperative game with a public good component. We establish existence of a unique non-trivial Nash equilibrium, and study its efficiency for a large class of disclosure and estimation cost functions. We also show an extension of the Aitken/Gauss-Markov theorem to this non-cooperative setup under certain conditions and examples in which the generalized least squares estimator is not optimal. 

Our Aitken/Gauss-Markov-type theorem is weaker than these two classical results in three ways. First, we proved the generalized least squares estimator is only approximately optimal in the case of a homogeneous estimation cost and of near-homogeneous disclosure cost and is not always optimal. Second, the optimality of the generalized least squares estimator in the case of monomial disclosure cost functions of same degree is shown w.r.t.~the homogeneous scalarization chosen, rather than the positive semidefinite order.
Finally, Theorem~\ref{GLSoptimaltofactor} applies to linear estimators whose difference from $\gls$ does not depend on the actions $\lambdab$. In the presence of arbitrary dependence on $\lambdab$, the non-trivial equilibrium need not be unique (or even exist). Understanding when this occurs, and proving optimality results in this context, also remains open.

Our model assumes that the precision chosen by each agent is known to the analyst. Amending this assumption brings issues of truthfulness into consideration: in particular, an important open question is whether there exists an estimator (viewed as a mechanism) that induces truthful precision reporting among agents, at least in equilibrium. An Aitken-type theorem seems instrumental in establishing such a result.  

Our analysis of the Nash equilibrium assumes complete information, that is that agents know the costs and features of other agents. Extending it to a Bayesian setting is an interesting open direction. Still, even when the complete information assumption does not hold, our present results are indicative in at least two ways. First, as our game is a potential game, we know that many natural dynamic evolutions of the game will converge to the Nash equilibrium. Second, if the number of agents grows large, we expect that the Nash equilibrium and Bayesian Nash equilibrium would be close since the empirical distribution of costs/features would then be close to the true underlying distribution.

\bibliographystyle{plain}
\bibliography{privacy_references,experimental_design_references}

\appendix
%!TEX root = techreport.tex

\section{Proof of Theorem~\ref{thm.Nash}}\label{proof.thm.Nash}

In this proof, we show that the potential function $\Phi$ is strictly
convex on its effective domain which implies that the set of Nash
equilibria that lie in its effective domain coincides with the set
of local minima of $\Phi$. By strict convexity, $\Phi$ has at most one
such local equilibrium. To conclude the proof, we then show that this
minimum is attained. % in the effective domain of $f$.

The potential function
$\Phi(\lambdab)=f(\lambdab)+\sum_ic_i(\lambda_i)$ takes values in the
extended positive real numbers line
$\bar{\R}_+=\R_+\cup\{+\infty\}$. By
Assumption~\ref{privacyassumption}, the disclosure costs $c_i(\cdot)$
are finite on $[0, 1/\sigma^{2}]$ since they are continuous on a
compact set. Therefore, $\Phi(\cdot)$ is finite if and only if $f(\cdot)$
is finite, i.e.,
$\dom \Phi \equiv \{ \lambdab: \Phi(\lambdab)<\infty\} = \dom f$,
where $\dom$ is the effective domain. Recall that since $X$ has rank
$d$, $(0,1/\sigma^{2}]^n\subseteq \dom \Phi$, and $\dom \Phi$ is
non-empty.

Recall that $V(\lambdab)=(X^T\Lambda X)^{-1}$. This implies that $V$ is
strictly convex and goes to infinity when $A(\lambdab)=X^T\Lambda X$
goes to a non-invertible matrix (i.e., the largest eigenvalue of $V$ goes to infinity for any sequence $\lambdab_n$ that converges to a $\lambdab$ such that $A(\lambdab)$ is non-invertible). As $F$ is convex and increasing, this
shows that $f(\lambdab)=F(V(\lambdab))$ is strictly convex and goes to
$+\infty$ when $A(\lambdab)$ goes to a non-invertible matrix, which then 
implies that $f(\lambdab):[0,1/\sigma^2]^n\to\bar{R}_+$ is
continuous. As the functions $c_i$ are convex, we conclude that the
potential function $\Phi$ is strictly convex and continuous on
$\bar{\R}_+$.

Let $B$ be the subset of $\lambdab$ such that
$\Phi(\lambdab)\le\Phi(1/\sigma^2\dots1/\sigma^2)$. By continuity and
convexity of $\Phi$, $B$ is a non-empty convex and compact subset of
$[0,1/\sigma^2]^n$ on which $\Phi(\lambdab)<\infty$. This implies that
the unique minimum of $\Phi$ is attained in $B\subseteq \dom(\Phi)$.
\qed

\section{Proof of Theorem~\ref{general}}\label{proof.general}

Under Assumptions~\ref{privacyassumption} and~\ref{estimationassumption}, the unique non-trivial equilibrium $\lambdabeq$ minimizes the potential function $\Phi(\lambdab)=\sum_{i\in N}c_i(\lambda_i)+f(\lambdab)$. Then, for $\lambdab^\opt$ a minimizer of the social cost:
$$\Phi(\lambdabeq)\leq \Phi(\lambdab^\opt)= \sum_{i\in N}c_i(\lambda_i^\opt)+f(\lambdab^\opt) \leq \sum_{i\in N}c_i(\lambda_i^\opt)+nf(\lambdab^\opt)=\opt $$
by the positivity of $f$. On the other hand, $\SC(\lambdabeq)\leq n\Phi(\lambdabeq)$, by the positivity of $c_i$, and the theorem follows.
\qed

%%%%%%%%%%%%%%%%%%%%%%%%%%%%%%%%%%%%%%%%%%%%%%%%%
% Proof of PoS thm

\section{Proof of Theorem~\ref{generalizedBound}}\label{proofofthmgeneralized}

To simplify the notation, in this proof, we write $p$ instead of $\pone$; hence we show that $\pos \le n^{\frac{q}{p + q}}$.
 
\noindent\textbf{Upper Bound.}  Recall that Assumption~\ref{homogeneityc} implies that $\forall \lambda \in \R_+, \forall a \ge 1, a^{p}c_i(\lambda) \le c_i(a\lambda)$.  This implies, by rewriting the assumption with $\lambda^{\prime} = a\lambda$, that $c_i(\frac{\lambda^{\prime}}{a}) \le a^{-p}c_i(\lambda^{\prime})$ for all $a \ge 1$ and for all $\lambda^{\prime}$. 

Recall that we denote by $\lambdabeq$ the unique non-trivial Nash equilibrium. Suppose that $\pos > n^{\frac{q}{p + q}}$, that is 
\begin{align*}
  \sum_{i\in N} c_i(\lambdaeq_i) + nf(\lambdabeq) &> n^{\frac{q}{q + p}} (\sum_{i\in N} c_i(\lambda_i^\opt) + nf(\lambdab^\opt)).
\end{align*}
We will show that this implies that $\lambdabeq$ is not an equilibrium, which is a contradiction. 

By using that $c_i(\lambdaeq_i)\ge0$ and dividing the above inequality
by $n$, we obtain:
\begin{align*}
  \sum_{i\in N} c_i(\lambdaeq_i) + f(\lambdabeq)
  &\ge \frac1n\left(\sum_{i\in N} c_i(\lambdaeq_i) + nf(\lambdabeq)\right) \\
  &> n^{-\frac{p}{q + p}}\sum_{i\in N} c_i(\lambda_i^\opt) + n^{\frac{q}{p + q}}f(\lambdab^\opt)\\
  &\ge \sum_{i\in N} c_i\left(\frac{\lambda_i^\opt}{n^{\frac{1}{p + q}}}\right) + 
    f\left(\frac{\lambdab^\opt}{n^{\frac{1}{p + q}}}\right),
\end{align*}
where for the last inequality, we used Assumption~\ref{homogeneityc}
and Assumption~\ref{homogeneityf} with $a=n^{1/(p+q)}$.

To conclude the proof, we remark that
$\frac{\lambdab^\opt}{n^{1/(p + q)}} \le \lambdab^\opt$ which implies
that $\frac{\lambdab^\opt}{n^{1/(p + q)}}$ is a valid strategy profile. This
would imply that $\lambdabeq$ is not the minimum of the potential
function which is a contradiction. Thus, we have
$\pos \le n^{\frac{q}{p + q}}$.\qed
\medskip

\noindent\textbf{Lower Bound.}  Fix $p\geq 1$ and $q \geq 1$. We consider a
$1$-dimensional model ($d=1$) with $x_1=1$ and
$\sigma^2=(q/p)^{1/(p+q)}$.  Let $c_i(\lambda_i) = \lambda_i^p$ for
all $i$ and $F(V) = \trace(V)^q=V^q$ (the last equality holds
because when $d=1$, the co-variance matrix is a scalar).  Hence, the
co-variance matrix is $V(\lambdab)=(\sum_{i\in N}\lambda_i)^{-1}$.

As all agents are identical, and by uniqueness of the Nash equilibrium, the Nash equilibrium is a symmetric Nash equilibrium where all agents will give the same value $\lambda^*$ where $\lambda^*$ is the unique minimizer of the potential function:
\begin{align*}
  n\lambda^p + (n\lambda)^{-q}.
\end{align*}
The minimum of this function is attained when its derivative is equal
to $0$. This implies that $np\lambda^{p-1}=nq (n\lambda)^{-q-1}$ which
implies that $\lambda^{p+q}=(q/p)n^{-1-q}$. This shows that
$\lambda^*=((q/p)n^{-1-q})^{1/(p+q)}$.

Similarly, the socially optimal $\lambdab^{\opt}$ is also symmetric and is attained when all agents give $\lambda^{\opt}$ the unique minimizer of the social cost:
\begin{align*}
  n\lambda^p + n(n\lambda)^{-q}.
\end{align*}
This implies that
\begin{equation}
\label{eq.loptlstar}
\lambda^\opt=(n(q/p)n^{-1-q})^{1/(p+q)}=n^{1/(p+q)}\lambda^*.
\end{equation}
Hence, we get:
\begin{align*}
  \pos &= \frac{C(\lambdabeq)}{C(\lambdab^\opt)} = \frac{ n (\lambda^*)^p+
         n(n\lambda^*)^{-q}}{ n (\lambda^\opt)^p+
         n(n\lambda^\opt)^{-q}}\\
       &= \frac{(\lambda^*)^p+ (n\lambda^*)^{-q}}{ (\lambda^\opt)^p+
         (n\lambda^\opt)^{-q}}\\
       &=
         \frac{(n\lambda^*){-q}}{(n\lambda^{\opt})^{-q}}\frac{(\lambda^*)^{p+q}+1}{(\lambda^\opt)^{p+q}+1}\\
       &=
         \left(\frac{\lambda^\opt}{\lambda^*}\right)^q\frac{1+(\lambda^*)^{p+q}}{1+(\lambda^\opt)^{p+q}}\\
       &=n^{q/(p+q)} \frac{ 1 + (q/p)n^{-1-q}} { 1 + (q/p)n^{-q}}, 
\end{align*}
where we use the expression in \eqref{eq.loptlstar} for $\lambda^*$ and $\lambda^{\opt}$ in
the last line. This shows that, for any $\epsilon$, for large enough $n$, the price of stability
is at least $n^{q/(p+q)}(1-\varepsilon)$.  \qed

\section{Proof of Theorem~\ref{GLSoptimaltofactor}}\label{GLSOptimal}

\subsection{Proof of (i)}

We consider the same setting as Example~1, i.e., a $1$-dimensional
model ($d=1$) with two agents in which the public data of each agent
is $x_i=1$. For such a game, the $\gls$ estimator is
$(X^T\Lambda X)^{-1}X^T\Lambda \tilde{\yb}=
(\lambda_1+\lambda_2)^{-1}\lambdab^T \tilde{\yb}$ and its covariance
is $1/(\lambda_1+\lambda_2)$. We consider a linear estimator
$\hat{\betab}(\delta)$ with $\delta \ge 0$ of the form
$\hat{\betab}_\gls+ \deltab^T \tilde{\yb}$ where
$\deltab\in\mathbb{R}^2$ is a vector with coefficients
$\delta_1 = -\delta_2 = \sqrt{\delta}$. Note that 
$\delta_1 = -\delta_2$ guarantees that this linear estimator is
unbiased.  We assume that the disclosure cost of
Agent~1 is $c_1(\lambda) = \lambda^{p_1}$
while the disclosure cost of Agent~2 is
$c_2(\lambda) = \lambda^{p_2}$. For a given $\delta$, we
  denote the equilibrium of the game by $\lambdabeq(\delta)$.
  
  Overall, this proof is decomposed in two steps:
\begin{quote}
  Step~1: We compute the derivative of the estimation cost at
  $\delta=0$ to show that it is negative if and only if
  $\lambdaeq_1(0)(2p_1-p_2 - p_1p_2) +
  \lambdaeq_2(0)(2p_2-p_1-p_1p_2)>0$.
\end{quote}
\begin{quote}
  Step~2: We show that there exists $x>0$ such that the above
  inequality is satisfied for $p_1=1/x$ and $p_2=1+x$.
\end{quote}

We describe both steps in detail below.

\paragraph{Step~1.} According to \eqref{eq:cov}, the covariance of the
estimator is
$1/(\lambda_1+\lambda_2) + \delta/\lambda_1+\delta/\lambda_2$, where
$\delta/\lambda_1+\delta/\lambda_2$ is the loss of precision due to
using a linear estimator that is less precise than $\gls$. We assume
that the scalarization function is the identity, which means that the
estimation cost is
\begin{align}
  \label{eq:f_L}
  f_{\delta}(\lambdab) = \frac1{\lambda_1+\lambda_2} + \frac\delta{\lambda_1}+\frac\delta{\lambda_2}.
\end{align}
The equilibrium $\lambdabeq(\delta)$ is the minimum of the
  potential function
  $\Phi_\delta(\lambdab)=f_\delta(\lambdab)+\lambda_1^{p_1}+\lambda_2^{p_2}$.
  The estimation cost at equilibrium is
  $f_\delta(\lambdabeq(\delta))$. Our goal in this step is to compute
  the derivative of $f_\delta(\lambdabeq(\delta))$ with respect to $\delta$ and to
  obtain a condition that ensures that it is negative at $\delta=0$.
  Let use denote by $(\lambdaeq)'_i(\delta)=d\lambdaeq_i(\delta)/(d\delta)$
  the derivative of $\lambdaeq_i(\delta)$ with respect to $\delta$.  To
  simplify notation, we will omit the dependence on $\delta$ and
  simply denote $\lambdaeq_i=\lambdaeq_i(0)$ and
  $\lambda_i'=(\lambdaeq)'_i(0)$ when it is not confusing.  The derivative
  of the estimation cost evaluated at $\delta=0$ is equal to
  \begin{align}
    \label{eq:der_estimation_cost}
    \frac{d}{d\delta}(f_\delta(\lambdabeq(\delta)))\Big|_{\delta=0}
    &=-\frac{\lambda'_1+\lambda'_2}{(\lambdaeq_1+\lambdaeq_2)^2}+\frac1{\lambdaeq_1}+\frac1{\lambdaeq_2}
      =-\frac{\lambda'_1+\lambda'_2}{(\lambdaeq_1+\lambdaeq_2)^2}+\frac{\lambdaeq_1+\lambdaeq_2}{\lambdaeq_1\lambdaeq_2}. 
    % \\
  % &=\frac{\lambda_1+\lambda_2}{\lambda_1\lambda_2}\left[
  %   -\lambda_1\lambda_2\frac{\lambda'_1+\lambda'_2}{(\lambda_1+\lambda_2)^3}+1\right].
  \end{align}
  In particular, the above derivative is negative if an only if
  $\frac{\lambda'_1+\lambda'_2}{(\lambdaeq_1+\lambdaeq_2)^3}\lambdaeq_1\lambdaeq_2>1$.
  In what follows, we compute the derivatives $\lambda_i'$ as a
  function of the values of $\lambdaeq_i$ and $p_i$.

The equilibrium $\lambdabeq(\delta)$ is the minimum of the potential
function
$\Phi_\delta(\lambdab)=f_\delta(\lambdab)+\lambda_1^{p_1}+\lambda_2^{p_2}$. By
using the first order condition
$\partial\Phi_\delta/\partial\lambda_i=0$, this implies that for all
$\delta\ge0$:
\begin{align}
  \label{eq:derivative}
  -\frac1{(\lambdaeq_1(\delta)+\lambdaeq_2(\delta))^2} -
  \frac\delta{(\lambdaeq_i(\delta))^2} + 
  p_i(\lambdaeq_i(\delta))^{p_i-1} = 0, \quad \text{for}~i\in\{1,2\}.
\end{align}
% Recall that $\delta=0$ corresponds to $\gls$ and that therefore,
% $\lambdabeq_{\gls}=\lambdabeq(0)$ is the precision at equilibrium.

The derivative of $\lambdaeq_i(\delta)$ with respect to $\delta$
exists thanks to the implicit function theorem.  By
differentiating \eqref{eq:derivative} with respect to $\delta$, we
obtain
\begin{align}
  0 &= \frac{d}{d \delta}\left(
      -\frac1{(\lambdaeq_1(\delta)+\lambdaeq_2(\delta))^2} - \frac\delta{(\lambdaeq_i(\delta))^2} +
      p_i(\lambdaeq_i(\delta))^{p_i-1}\right)\nonumber\\
    &=  2\frac{(\lambdaeq)'_1(\delta)+(\lambdaeq)'_2(\delta)}{(\lambdaeq_1(\delta)+\lambdaeq_2(\delta))^3}
      -
      \frac1{(\lambdaeq_i(\delta))^2}+2\delta\frac{(\lambdaeq)'_i(\delta)}{(\lambdaeq_i(\delta))^3}+p_i(p_i-1)(\lambdaeq_i(\delta))^{p_i-2}(\lambdaeq)'_i(\delta).
      \label{eq:der1}
\end{align}
Equation~\eqref{eq:derivative}, evaluated at $\delta=0$, shows that
$p_i(\lambdaeq_i)^{p_i-1}=\frac{1}{(\lambdaeq_1+\lambdaeq_2)^2}$. Evaluating
Equation~\eqref{eq:der1} at $\delta=0$ and plugging this equality
gives
\begin{align}
  0&= 2\frac{\lambda'_1+\lambda'_2}{(\lambdaeq_1+\lambdaeq_2)^3}
     -
     \frac1{(\lambdaeq_i)^2}+p_i(p_i-1)(\lambdaeq_i)^{p_i-2}\lambda_i'\nonumber
  \\
   &=2\frac{\lambda'_1+\lambda'_2}{(\lambdaeq_1+\lambdaeq_2)^3} 
     -
     \frac1{(\lambdaeq_i)^2}+\frac1{(\lambdaeq_1+\lambdaeq_2)^2}\frac{p_i-1}{\lambdaeq_i}\lambda_i'.
      \label{eq:der}
\end{align}
In order to isolate the term $\lambda_1'+\lambda_2'$, we multiply the
above equation by $\lambdaeq_i/(p_i-1)$ and we sum over
$i\in\{1,2\}$. This gives:
\begin{align*}
  0 &= 2\frac{\lambda_1'+\lambda_2'}{(\lambdaeq_1+\lambdaeq_2)^3}\left(\frac{\lambdaeq_1}{p_1-1}+\frac{\lambdaeq_2}{p_2-1}\right)
  - \frac1{\lambdaeq_1(p_1-1)}-\frac1{\lambdaeq_2(p_2-1)} +
      \frac{\lambda'_1+\lambda'_2}{(\lambdaeq_1+\lambdaeq_2)^2}
  \\
    &=
      \frac{\lambda_1'+\lambda_2'}{(\lambdaeq_1+\lambdaeq_2)^3}\left(\frac{2\lambdaeq_1}{p_1-1}+\frac{2\lambdaeq_2}{p_2-1}+\lambdaeq_1+\lambdaeq_2\right) 
      - \frac1{\lambdaeq_1(p_1-1)}-\frac1{\lambdaeq_2(p_2-1)}.
\end{align*}
This shows that
\begin{align*}
  \frac{\lambda_1'+\lambda_2'}{(\lambdaeq_1+\lambdaeq_2)^3} &= \frac{
  \frac1{\lambdaeq_1(p_1-1)} +
  \frac1{\lambdaeq_2(p_2-1)}}{\frac{2\lambdaeq_1}{p_1-1}+\frac{2\lambdaeq_2}{p_2-1}+\lambdaeq_1+\lambdaeq_2}\\
  &=
    \frac{1}{\lambdaeq_1\lambdaeq_2}\frac{\lambdaeq_2(p_2-1)+\lambdaeq_1(p_1-1)}{2\lambdaeq_1(p_1-1)+2\lambdaeq_2(p_2-1)
    + (\lambdaeq_1+\lambdaeq_2)(p_1-1)(p_2-1)}.
\end{align*}
In particular, this implies that the derivative
\eqref{eq:der_estimation_cost} is negative if and only if
\begin{align*}
  \frac{\lambdaeq_2(p_2-1)+\lambdaeq_1(p_1-1)}{2\lambdaeq_1(p_1-1)+2\lambdaeq_2(p_2-1)
  + (\lambdaeq_1+\lambdaeq_2)(p_1-1)(p_2-1)}>1.
\end{align*}
After some algebra, this gives 
\begin{align}
  \lambdaeq_1(2p_1-p_2 - p_1p_2) + \lambdaeq_2(2p_2-p_1-p_1p_2)>0,
  \label{eq:gls3}
\end{align}
where, again, by abuse of notation we denote
$\lambdaeq_1=\lambdaeq_1(0)$ and $\lambdaeq_2=\lambdaeq_2(0)$.

  \paragraph{Step~2.} We now consider $p_1=1+1/x$ and $p_2=1+x$ and
  $x\to\infty$.  To emphasize the dependence in $x$, let us denote by
  $\lambdabeq(x)=(\lambdaeq_1(x),\lambdaeq_2(x))$ the value of the precision
  at equilibrium (for $\gls$) and $\Phi_x(\cdot)$ the potential of the
  game. By definition, $\lambdabeq(x)$ minimizes
  $\Phi_x(\lambdab)=1/(\lambda_1+\lambda_2)+\lambda_1^{1+1/x}+\lambda_2^{1+x}$.
  This implies that for all $\epsilon>0$,
  $\Phi_x(\lambdabeq(x))\le\Phi_x(0,1-\varepsilon)$. As
  $\lim_{x\to\infty}\Phi_x(0,1-\varepsilon)=1/(1-\varepsilon)$ and
  because this is true for all $\varepsilon$, this implies that
\begin{align*}
  \lim_{x\to\infty}\Phi_x(\lambdabeq(x))=
  \lim_{x\to\infty}\left(\frac1{\lambdaeq_1(x)+\lambdaeq_2(x)}+(\lambdaeq_1(x))^{1+1/x}+(\lambdaeq_2(x))^{1+x}\right)\le1. 
\end{align*}
This implies that $\lim_{x\to\infty}\lambdaeq_1(x)=0$ and
$\lim_{x\to\infty}\lambdaeq_2(x)=1$.

For our values of $p_1=1+1/x$ and $p_2=1+x$, the left-hand
  side of \eqref{eq:gls3} equals
  $\lambda_1(x)(1/x-2x - 1) + \lambda_2(x)(x-2/x -1)$. As
  $\lim_{x\to\infty}\lambda_2(x)=1$ and
  $\lim_{x\to\infty}\lambda_1(x)=0$, this term is asymptotically
  equivalent to $x$ and is therefore positive for $x$ large enough.
This implies that there exists a value $x$ such that
$d/(d\delta)f_\delta(\lambdabeq(\delta))<0$.  Hence, for this $x$,
there exists a perturbation value $\delta>0$ such that
$\hat{\betab}(\delta)$ is an estimator that is more efficient than
$\gls$.  \qed

\subsection{Proof of (ii)}

We will start by proving Lemma~\ref{sum_cost_vs_common}.  This lemma
can easily be explained if we recall Assumption~\ref{homogeneityc} and
Assumption~\ref{homogeneityf}. Indeed, they dictate how the different
components of the potential function behave when all agents multiply
or divide the amount of information they give. If the sum of individual costs is
too great compared to the common cost then dividing the amount
that all agents give greatly reduces the individual costs while slightly
augmenting the common cost, which is beneficial. The same goes the
other way around where agents multiply the amount they give. This
formalizes an intuition that one can have about this model: there is a
balance between the individual costs paid to achieve the objective of reducing the common cost and the
objective itself.

\begin{lemma}\label{sum_cost_vs_common}
	Under Assumptions \ref{privacyassumption}, \ref{estimationassumption}, \ref{homogeneityc} and \ref{homogeneityf},
        the ratio between the sum of individual costs and the common cost is
        bounded. Formally, the equilibrium $\lambdabeq$ satisfies: 
	\begin{equation*}
	\sum_{i\in N} c_i(\lambdaeq_i) \le \frac{q}{\pone}f(\lambdabeq)
        \text{\qquad and \qquad}
        f(\lambdabeq) \le \frac{\ptwo}{q}\sum_{i\in N} c_i(\lambdaeq_i).
	\end{equation*}
\end{lemma}

\begin{proof}
  This proof mainly relies on the fact that $\lambdabeq$ is the
  minimum of the potential function. 
  Let $\lambdabeq$ be the unique non-trivial equilibrium.  Let
  $\kappa \in (0, 1)$ be a multiplicative factor applied to the
  equilibrium profile. As $\lambdabeq$ is the minimum of the potential
  function, we have $\Phi(\lambdabeq) \le \Phi(\kappa \lambdabeq)$ and
  $\Phi(\lambdabeq) \le \Phi(\lambdabeq/\kappa)$. This implies that:
  \begin{align*}
    \sum_{i\in N} c_i(\lambdaeq_i) + f( \lambdabeq)
    &\le \sum_{i\in N} c_i(\kappa \lambdaeq_i) + f(\kappa \lambdabeq)  
      \le \kappa^{\pone}  \sum_{i\in N} c_i(\lambdaeq_i) +
      \kappa^{-q}f( \lambdabeq),\qquad\text{and}
    \\
    \sum_{i\in N} c_i(\lambdaeq_i) + f( \lambdabeq)
    &\le \sum_{i\in N} c_i( \lambdaeq_i/\kappa) + f(\lambdabeq/\kappa)\le
      \kappa^{-\ptwo} \sum_{i\in N} c_i(\lambdaeq_i) + \kappa^{q}f( \lambdabeq).
  \end{align*}
   The above equations imply that:
    \begin{align*}
      (1-\kappa^{\pone})\sum_{i\in N} c_i(\lambdaeq_i) &\le
      (\kappa^{-q}-1) f( \lambdabeq), \qquad\text{and}\\
      (1-\kappa^q)f( \lambdabeq) &\le (\kappa^{-\ptwo}-1)\sum_{i\in N} c_i(\lambdaeq_i).
    \end{align*}
 As $\kappa\in(0,1)$, we have $1-\kappa^{\pone}>0$ and
 $\kappa^{-\ptwo}-1>0$. Hence, the above equations imply
  that
    \begin{align*}
      \frac{1-\kappa^{q}}{\kappa^{-\ptwo}-1} f(\lambdabeq) \le
      \sum_{i\in N} c_i(\lambdaeq_i)
      &\le\frac{\kappa^{-q}-1}{1-\kappa^{\pone}}f(\lambdabeq).
    \end{align*}
  This inequality is valid for every $\kappa \in (0, 1)$. As
  $\lim_{\kappa\to1}\frac{1-\kappa^{q}}{\kappa^{-\ptwo}-1}=q/\ptwo$
    and
    $\lim_{\kappa\to1}\frac{1-\kappa^{-q}}{1-\kappa^{-\pone}}=q/\pone$,
  this gives \begin{equation*} \frac{q}{\ptwo} f(\lambdabeq)\le
    \sum_{i\in N} c_i(\lambdaeq_i) \le \frac{q}{\pone} f(\lambdabeq).
  \end{equation*}
\end{proof}

We are now ready to prove Theorem~\ref{GLSoptimaltofactor}(ii).
Let $\Phi_L(\lambdab) = \sum_{i\in N} c_i(\lambda_i) + f_L(\lambdab)$ be the
potential function for any linear unbiased estimator and
$\Phi_{\gls}(\lambdab) = \sum_{i\in N} c_i(\lambda_i) + f_{\gls}(\lambdab)$ be
the potential function for $\gls$.
Recall that $\lambdabeq_L$ and $\lambdabeq_{\gls}$ denote the non-trivial
equilibria for the linear unbiased estimator and for $\gls$ respectively.  By optimality of $\gls$, for all
$\lambdab$ we have $f_L(\lambdab)\ge f_{\gls}(\lambdab)$. This implies
that for all $\lambdab$, we have
$\Phi_L(\lambdab) \ge \Phi_{\gls}(\lambdab)$. Therefore 
\begin{align}
  \label{eq:lambdabeq_L}
  \Phi_L(\lambdabeq_L)=\min_{\lambdab}\Phi_L(\lambdab) \ge
  \Phi_{\gls}(\lambdabeq_{\gls})=\min_{\lambdab}\Phi_{\gls}(\lambdab).
\end{align}

By applying the inequalities of Lemma \ref{sum_cost_vs_common}, we obtain: 
\begin{equation*}
\Phi_{\gls}(\lambdabeq_{\gls}) = \sum_{i\in N} c_i((\lambdabeq_{\gls})_i) + f_{\gls}(\lambdabeq_{\gls}) \ge \frac{q}{\ptwo} f_{\gls}(\lambdabeq_{\gls}) +  f_{\gls}(\lambdabeq_{\gls}),
\end{equation*}
and 
\begin{equation*}
\Phi_L(\lambdabeq_L) = \sum_{i\in N} c_i((\lambdabeq_L)_i) + f_L(\lambdabeq_L) \le \frac{q}{\pone} f_L(\lambdabeq_L) +  f_L(\lambdabeq_L).
\end{equation*}
Combining the above two inequalities with \eqref{eq:lambdabeq_L}, we conclude that
\begin{equation*}
f_{\gls}(\lambdabeq_{\gls}) \le \frac{\frac{q}{\pone} + 1}{\frac{q}{\ptwo} + 1} f_L(\lambdabeq_L) = \frac{\ptwo (q + \pone)}{\pone (q + \ptwo)} f_L(\lambdabeq_L).
\end{equation*}

\end{document}